\documentclass[reqno]{amsart}
\usepackage{amsmath,amssymb}
\usepackage{graphics,epsfig,color,psfrag,dsfont}
\usepackage[mathscr]{euscript}
\theoremstyle{plain}
\newtheorem{theorem}{Theorem}[section]
\newtheorem{proposition}[theorem]{Proposition}
\newtheorem{lemma}[theorem]{Lemma}

\theoremstyle{definition}

\theoremstyle{remark}
\newtheorem{remark}[theorem]{Remark}

\makeatletter
\newcommand*\nota[1][\mathord{*}]{%
	\xdef\@thefnmark{\ensuremath{\m@th#1}}\@footnotemark\@footnotetext}\makeatother
\numberwithin{equation}{section}
\numberwithin{theorem}{section}

\renewcommand{\epsilon}{\varepsilon}
\renewcommand{\tilde}{\widetilde}
\renewcommand{\hat}{\widehat}

\renewcommand{\div}{\mathop{\rm div}\nolimits}

\newcommand{\asinh}{\mathop{\rm arcsinh}\nolimits}

\definecolor{light}{gray}{.9}

 \newcommand{\be}{\begin{equation}}
\newcommand{\en}{\end{equation}}

\title[Current fluctuations for zero-range process on graphs]{Current fluctuations for the boundary-driven zero-range process on graphs: microscopic versus macroscopic approach and a theory of non-reversible resistor-like networks}

\author[D.\ Gabrielli]{Davide Gabrielli}
\address{Davide Gabrielli\footnote{Corresponding author} \hfill\break \indent
  DISIM, Universit\`a dell'Aquila
  \hfill\break\indent
  Via Vetoio,   67100 Coppito, L'Aquila, Italy}
\email{davide.gabrielli@univaq.it}

\author[R.\ J.\ Harris]{Rosemary J. Harris}
\address{Rosemary J. Harris \hfill\break \indent
Department of Mathematics, University College London \hfill\break\indent
Gower Street, London, WC1E 6BT, UK}
 \email{rosemary.j.harris@ucl.ac.uk}

\begin{document}

\begin{abstract}
We compute the joint large deviation rate functional in the limit of large time for the current flowing through
the edges of a finite graph on which a boundary-driven system of stochastic particles evolves with zero-range dynamics. This generalizes one-dimensional results previously obtained with different approaches \cite{BDGJL1,HRS}; our alternative techniques illuminate various connections and complementary perspectives. In particular, we here use a variational approach to derive the rate functional by contraction from a level 2.5 large deviation rate functional. We perform an exact minimization and finally obtain the rate functional as a variational problem involving a superposition of cost functions for each edge. The contributions from different edges are not independent since they are related by the values of a potential function on the nodes of the graph. The rate functional on the graph is a microscopic version of the continuous rate functional predicted by the macroscopic fluctuation theory \cite{MFT}, and we indeed show a convergence in the scaling limit. If we split the graph into two connected regions by a cutset and are interested just in the current flowing through the cutset, we find that the result is the same as that of an effective system composed of only one effective edge (as happens at macroscopic level and is expected also for other models \cite{Cap}).  The characteristics of this effective edge are related to the ``capacities'' of the graph and can be obtained by a reduction using elementary transformations as in electrical networks; specifically, we treat components in parallel, in series, and in $N$-star configurations (reduced to effective complete $N$-graphs). Our reduction procedure is directly related to the reduction to the trace process \cite{L} and, since the dynamics is in general not reversible, it is also closely connected to the theory of non-reversible electrical networks in \cite{B}.

\bigskip


\smallskip


\end{abstract}

\maketitle
\thispagestyle{empty}

\section{Introduction}
An important challenge in non-equilibrium statistical mechanics is understanding the structure of current fluctuations through a system.  In common with other problems in statistical mechanics, the study of current fluctuations can be tackled using several different approaches, with a particular distinction between those having a microscopic perspective and those having a macroscopic perspective.

In this paper we study the large deviation rate functional for the current of a boundary-driven zero-range system of particles evolving on a finite lattice and in contact with external sources. Particles jump through the edges between the nodes of a finite graph, and are created or destroyed at boundary nodes.  There is a zero-range interaction among the particles themselves: the total jump rate of each particle is affected by the number of particles that are present at the same vertex,  with the restriction that the total jump rate from every node grows at least linearly with the number of particles in order to avoid the phenomenon of condensation~\cite{Gros,LMS,GSS,CC}. We compute the joint large deviation rate functional for the current of particles across each edge of the graph. This is the main result of the paper. Large deviations for the current of zero-range dynamics have been previously obtained, for example in \cite{BDGJL1, bd,HRS}, but our result may be the first one that explicitly gives a joint large deviation principle (LDP) for all the edges.
A special feature of the rate functional is that it does not depend on the specific form of the zero-range interaction; this is a feature already observed  in \cite{BDGJL1, bd,HRS} and it merits further comment later in the present paper.

We obtain the joint LDP by applying the contraction principle to a level 2.5 LDP \cite{BaCh,BFG} on the configuration space. We can solve exactly the corresponding variational problem to get a cost function on the graph which involves the currents across the edges and some potentials associated to the nodes. The form of the cost function is a sum of terms, one for each edge, like the energy of a resistor network.

The problem of large deviations for the current can be tackled, in some cases, also with a different  macroscopic approach: the so-called macroscopic fluctuation theory (MFT) \cite{MFT,BDGJL1}. One case where this approach can be used is the large deviations of symmetric zero-range models in the diffusive hydrodynamic scaling limit. In the large-time limit and in the absence of dynamical phase transitions, the corresponding continuous rate functional is a variational problem involving the current, an integral over the domain and a minimization over a potential. Our discrete rate functional, in the symmetric case, is a discrete version of the rate functional deduced from the macroscopic fluctuation theory. Indeed our microscopic rate functional converges to the continuous one in the limit when the mesh of the lattice is taken to zero. This is our second main result. Moreover, in the one-dimensional case, we find exact correspondences between additivity formulas and between minimizing potentials in the discrete and continuous settings.

With both microscopic and macroscopic approaches in the one-dimensional case, it is possible to obtain that the large deviation rate functional reduces to the rate functional of a graph with one single effective edge. This is obtained in the discrete setting by a reduction similar to the reduction of resistors in series.

In the paper \cite{Cap} the problem of fluctuations of the current across a cutset separating a system into two disjoint subsystems was studied.  For the simple exclusion process this problem was considered both macroscopically, for the rate functional of the macroscopic fluctuation theory, and also, heuristically, microscopically. The authors deduced that in both cases the fluctuations are equivalent to those of an effective one-dimensional system.

For the zero-range model, we can prove the above equivalence exactly, supplementing the one-dimensional reduction of the edges in series with a reduction procedure for edges in parallel, and obtaining also a star-triangle reduction that can be generalized to stars with $N$ vertices and complete $N$-graphs. Since our models are not necessarily reversible, we obtain a reduction procedure for the graph that is tantamount to the non-reversible electrical theory for non-reversible random walks proposed in \cite{B}. The discrete large deviation rate functional for the current here plays the role of a non-quadratic energy. This is our third main result. The classic reference for the relation between stochastic models and resistor networks is \cite{DS}. More recently there are several papers applying techniques from resistor networks in the framework of non-equilibrium systems, see for example \cite{A,C-electric,lin,Z}.  In particular, after completing the present work we discovered that our own network reduction procedure, with the same cost function, has also been obtained in \cite{PS}; however, the authors of \cite{PS} consider only reversible rates, with proofs and motivations which are rather different to ours.

We concentrate our exposition on the more computational aspects of the rate functional, treating some of the other issues without complete rigour in order to have a paper of reasonable length.

\smallskip

The work is organized as follows.

In Section 2 we introduce the class of models under consideration and discuss some basic issues and notation.

In Section 3 we recall the statement of level 2.5 large deviations and explain how this can be used to obtain the joint LDP for the current on the graph.

In Section 4 we demonstrate how to solve the variational problem of the contraction principle and thus arrive at the expression of the joint LDP.

In Section 5 we discuss the one-dimensional model in full generality and show the discrete additivity property.

In Section 6 we compute the scaling limit in the symmetric case and compare the results with the rate functional from macroscopic fluctuation theory.

In Section 7 we show the reduction rules of a general graph with the introduction of effective edges and discuss the relation with the trace process.

In Section 8 we compute the large deviations for the current through a cutset and show that this is equivalent to a graph with one single edge and effective rates.

In the appendices we collect some technical results.

\section{Boundary-driven zero-range process} \label{sec:1}

\subsection{Model definition}
Consider a finite directed graph $(\Omega_N\cup \partial \Omega_N, E_N)$ where the sets $\Omega_N$ and $\partial \Omega_N$ are disjoint. The nodes in $\partial \Omega_N$ are ``ghost'' sites associated to creation and annihilation of particles, physically they represent the external sources (or ``reservoirs'')  in contact with the system and we call them also ``boundary'' sites in the following. The nodes in $\Omega_N$ are the nodes of the system where particles evolve. The parameter $N$ can be thought of as fixed and related to the size of the system. In Section \ref{sec:scaling} we will consider the limit $N\to+\infty$ and require that $|\Omega_N|$ is diverging with $N$. 

We define
$\Lambda_N:=\left(\mathbb N \cup
\left\{0\right\}\right)^{\Omega_N}$. An element $\eta\in \Lambda_N$
is interpreted as a configuration of particles in $\Omega_N$ with $\eta(x)$
particles at site $x\in \Omega_N$. We denote by $\delta_x \in \Lambda_N$
the configuration of particles containing one single particle at $x$. Defining
$\eta^{x,y}:=\eta-\delta_x+\delta_y$ and $\eta^{x,\pm}:=\eta\pm \delta_x$, we have that $\eta^{x,y}$ and $\eta^{x,-}$
belong to $\Lambda_N$ if and only if $\eta(x)>0$; when this happens then $\eta^{x,y}$ is the
configuration of particles obtained from $\eta$ by letting one particle at $x$ jump to $y$, while $\eta^{x,\pm}$ are
configurations of particles obtained from $\eta$ by creating or destroying one particle at $x$. Let
$g:\mathbb N\cup\{0\}\to \mathbb R^+$ be a function such that $g(0)=0$.  Consider also a weight function $c: E_N\to \mathbb R^+$ and finally a function $\lambda:\partial \Omega_N\to \mathbb R^+$ that determines the rate at which a ghost site in $\partial \Omega_N$ produces particles to be injected into the system $\Omega_N$.
The boundary
driven zero-range process is a continuous-time Markov chain on
$\Lambda_N$ depending on the function $g$, the weights $c$, and the function $\lambda$,
and determined by the following jump rates
\begin{align}\label{gen-zr}
\mathcal L_Nf(\eta)=&\sum_{\{x,y\in \Omega_N, (x,y)\in  E_N\}}g(\eta(x))c(x,y)\left(f(\eta^{x,y})-f(\eta)\right) \nonumber \\
&+\sum_{\{x\in \Omega_N,y\in \partial \Omega_N,(x,y)\in E_N\}}g(\eta(x))c(x,y)\left(f(\eta^{x,-})-f(\eta)\right) \nonumber \\
&+\sum_{\{y\in \Omega_N,x\in \partial \Omega_N,(x,y)\in E_N\}}\lambda_xc(x,y)\left(f(\eta^{y,+})-f(\eta)\right)\,.
\end{align}

We now give the informal description of this Markov chain on the countable state space $\Lambda_N$.
Particles are injected into the discrete domain $\Omega_N$ in the following way. For any $(x,y)\in E_N$ such that $x\in \partial \Omega_N$ and $y\in \Omega_N$, a particle is created at $y$ with rate $\lambda_xc(x,y)$.
Inside the lattice each particle performs
a continuous-time random walk with rate given by the function $c$ multiplied by an intensity which is dependent on the number of particles
present at the site where the particle is located. This change of speed is the only interaction
among the particles and it is determined by the function $g$. Finally, with rate $g(\eta(x))c(x,y)$, when $x\in \Omega_N$, $y\in \partial \Omega_N$ and $(x,y)\in E_N$, one particle at $x$ jumps to the ghost site $y$ and is destroyed.

\smallskip

Significantly, although the zero-range model is a many-particle process, we will see in the following that useful information about the current fluctuations can be obtained by considering the continuous-time random walk on $\Omega_N\cup\partial \Omega_N$ with transition rates given by the parameters $\left(c(x,y)\right)_{(x,y)\in  E_N}$, absorbed at $\partial \Omega_N$ and having injection rates given by $(\lambda_x)_{x\in \partial \Omega_N}$.

\smallskip

Throughout this work we assume the irreducibility condition that a walker following the above random walk created at any $x\in \partial \Omega_N$ can reach any $y\in \Omega_N$ and be absorbed on any $z\in \partial \Omega_N$.

\subsection{Invariant measures}
We will assume that the function $g$ is superlinear, i.e.\ there exists a positive constant $A$ such that $g(k)\geq Ak$. This assumption on $g$ is a sufficient condition to get our large deviation results; significantly, it allows a direct comparison with the independent-particle case as outlined in Appendix \ref{appA}. We believe that our results might also hold under less restrictive assumptions on the growth of the function $g$ but we do not discuss that point further here. As pointed out by an anonymous reviewer, the results should also hold in the case of zero-range dynamics with site-dependent superlinear functions $g_x$. Again, we do not discuss the details of such an extension. 

\smallskip

Under the superlinear assumption we avoid possible condensation phenomena \cite{Gros,LMS} and guarantee that the model has a unique invariant measure.
The invariant measure is of product type \cite{BDGJL1, EH}
and can be written as
\begin{equation}\label{prod-zr}
\mu[\phi]:=\prod_x\frac{\phi_x^{\eta(x)}}{Z(\phi_x)g(\eta(x))!}\,,
\end{equation}
where $g(k)!:=g(k)g(k-1)\cdots g(1)$ and $Z(\phi):=\sum_{k=0}^{+\infty}\frac{\phi^k}{g(k)!}$ is a normalization factor.
The function $\phi:\Omega_N\cup \partial \Omega_N\to \mathbb R_+$ is determined by the parameters of the model and in particular has to satisfy the following discrete set of equations
\begin{equation}\label{inv-gen}
\left\{
\begin{array}{ll}
\phi_x=\lambda_x\,, & x\in \partial \Omega_N \\
\sum_{y: (x,y)\in  E_N}\phi_xc(x,y)-\sum_{y: (y,x)\in E_N}\phi_yc(y,x)=0\,, & x\in \Omega_N\,.
\end{array}
\right.
\end{equation}
The above form of the invariant measure can be checked by a direct computation and we refer to
\cite{BDGJLprimo,EH,LMS} for the details.
As a special case we have the one-dimensional lattice for which $\Omega_N:=\left\{1,2,\dots N\right\}$
and $\partial \Omega_N=\left\{0, N+1\right\}$; here the set of equations \eqref{inv-gen} becomes
\begin{equation}\label{inveqc}
\left\{
\begin{array}{ll}
\phi_0=\lambda_L\,, & \\
\phi_{N+1}=\lambda_R\,, & \\
\sum_{a=\pm 1}\Big[\phi_xc(x,x+a)-\phi_{x+a}c(x+a,x)\Big]=0\,, & x=1,\dots ,N\,
\end{array}
\right.
\end{equation}
with boundary rates $\lambda_L := \lambda_0$ and $\lambda_R := \lambda_{N+1}$.
For more details see~\cite{EH, Pulk} and note that this setting includes, as a special case, the homogeneous model with bulk driving, i.e.\ $c(x,x+1)$ constant but different to $c(x+1,x)$, which is treated in~\cite{LMS}.

\section{Level 2.5 large deviations and current fluctuations}

We discuss here briefly a widely-applicable large deviation principle for Markov processes which is commonly called level 2.5 large deviations. The reader is referred to \cite{BaCh,BFG} for details and historical development of the framework, as well as to \cite{BFG3,BFG2,MN,Monthus} for examples and applications. Our aim is to apply this large deviation principle to the general case of the boundary-driven zero-range process.

\subsection{Empirical measure and flow}\label{empmf}
Given $(\eta_t)_{t\in \mathbb R^+}$ a sample path of a continuous-time Markov chain (in our case, the zero-range process), i.e.\ an element  $D(\mathbb R^+, \Lambda_N)$ with jumps compatible with the transition graph, we define
the associated empirical flow and measure. For any time $T>0$, the empirical measure $\mu_T$ is a probability
measure on $\Lambda_N$. The empirical flow $Q_T$ is a flow on the transition graph $(\Lambda_N, \tilde E_N)$, i.e.\ it is a positive function defined on the directed edges $(\eta,\eta')\in \tilde E_N$ where $\eta'$ is obtained from $\eta$  via an elementary
transition of the chain. Both $\mu_T$ and $Q_T$ depend on the trajectory up to time $T$, that is $(\eta_t)_{\in \mathbb [0,T]}$, and are defined as follows.

The empirical measure quantifies the fraction of time spent on each element of the state space and is given by
\begin{equation}\label{empmes}
\mu_T(\eta):=\frac 1T\int_0^T\delta_{\eta,\eta_s}ds\,, \qquad \eta\in \Lambda_N\,,
\end{equation}
where $\delta_{\cdot,\cdot}$ is the Kronecker delta.
The empirical flow is related to the number of jumps across the edges of the transition graph and is defined as
\begin{equation}\label{empfl}
Q_T(\eta,\eta'):= \frac{\textrm{Number\ of\ jumps\ from} \ \eta \ \textrm{to}\ \eta'\ \textrm{in}\ [0,T]}{T}\,, \qquad (\eta,\eta') \in \tilde E_N\,.
\end{equation}

Since we also want to take into account the amount of mass that is flowing across edges $(x,y)\in E_N$ where one of $x$ and $y$ belongs to $\Omega_N$ while the other belongs to $\partial \Omega_N$, and since there may be multiple such boundary edges for a given internal node, we consider the transition graph $(\Lambda_N, \tilde E_N)$ as a multigraph. This is necessary only in the case when $(\eta,\eta')$ is such that $\eta'=\eta^{x,\pm}$ for some $x$. For example if $\eta'=\eta^{x,+}$ and there are several $y\in \partial \Omega_N$ such that $(y,x)\in E_N$, then we consider in $(\Lambda_N, \tilde E_N)$ as many directed edges from $\eta$ to $\eta^{x,+}$ as the number of $y\in \partial \Omega_N$. The process jumps across the edge that corresponds to a particular such $y$ with rate $\lambda_yc(y,x)$.
In this way we can recover from the empirical flow the amount of mass flowing through each edge in $E_N$, including the boundary ones.

For convenience we now extend our notation to write $\eta^{x,y}$ even when one of $x$ and $y$ does not belong to $\Omega_N$.  In particular, we define $\eta^{x,y}=\eta^{y,+}$ if $x\in \partial \Omega_N$ and $y\in \Omega_N$, and $\eta^{x,y}=\eta^{x,-}$ if $x\in \Omega_N$ and $y\in \partial\Omega_N$.
For example, with this generalized notation, we can denote naturally as $Q(\eta,\eta^{x,y})$ the number
of particles created at $y\in \Omega_N$ that have been generated at $x\in \partial \Omega_N$, starting from state $\eta$, and then moved through the edge $(x,y)$.

\subsection{Level 2.5 large deviations}
\label{sec2point5}

The level 2.5 large deviation principle is a joint large deviation result for the pair $(\mu_T,Q_T)$
in the limit $T\to +\infty$. We have therefore informally that
\begin{equation}\label{palloncini}
\mathbb P_{\eta}\left((\mu_T,Q_T)\sim (\mu,Q)\right)\simeq e^{-TI(\mu,Q)}
\end{equation}
where $I$ is the rate functional which we seek to describe. Note that \eqref{palloncini} introduces the shorthand notation used for large deviations throughout this work: the symbol $\sim$ denotes closeness in a suitable metric while $\simeq$ denotes the validity of upper and lower bounds, as in the classic definition \cite{DZ}.

We concentrate here on the computational power of this LDP and refer to~\cite{BFG, BFG2} for a complete mathematical description.  The precise formulation of the LDP and the application of the contraction principle is non-trivial for the zero-range process since the transition graph is infinite.  However, in the particular case of superlinear growth of the rates, a coupling with independent particles (as outlined in Appendix~\ref{appA}) shows that an LDP with the empirical flow as an element of $L^1_+(\tilde E_N)$ is applicable and the contraction principle too. Here $L^1_+(\tilde E_N)$ is the space of positive summable flows on the edges $\tilde E_N$ with the strong topology.  We refer to \cite{BFG2} for the precise formulation as well as discussion of the conditions for the validity of an LDP in $L^1_+$ space
  and of related topological issues.

In formula \eqref{palloncini} we denote by $\mathbb P_\eta$ the probability on path space when the initial condition of the process is $\eta$ although, in fact, for superlinear $g$ the rate functional does not depend on $\eta$. It does for other choices of $g$, when one has ``instantaneous condensation''~\cite{HRS} with the associated possibility of seeing flows which are not summable and/or not divergence-free; we will not consider such cases and our flows will always be summable and divergence-free, as explained below.

The rate functional $I=I(\mu,Q)$ in \eqref{palloncini} depends on the pair $(\mu, Q)$ where $\mu$ is a probability measure on $\Lambda_N$ and $Q$ is a flow on $\tilde E_N$ [Recall that a flow is a map from the set of directed edges to $\mathbb R^+$]. Obviously, the rate functional depends on $N$; for simplicity, however, we suppress such dependence.   Under our assumption on $g$, the flow $Q$ has to satisfy the following two conditions in order that $I(\mu,Q)<+\infty$.
\begin{itemize}
	\item \emph{Summability}: We have that $I(\mu,Q)=+\infty$
unless the flow $Q$ is summable; we call here a flow $Q$ on $(\Lambda_N,\tilde E_N)$  \emph{summable} if $\sum_{(\eta,\eta')\in \tilde E_N}Q(\eta,\eta')<+\infty$.

\item \emph{Divergence-free}: We have that $I(\mu,Q)=+\infty$ unless the flow $Q$ is divergence-free, i.e.\ satisfies
\begin{equation}
\nabla\cdot Q(\eta):=\sum_{\eta'\,:\, (\eta,\eta')\in \tilde E_N}Q(\eta,\eta')-\sum_{\eta'\,:\, (\eta',\eta)\in \tilde E_N}Q(\eta',\eta)=0\,,
\label{zero-div-zr}
\end{equation}
where, to avoid confusion, we emphasize that we have defined here the discrete divergence using the same symbol as for its continuous analog.
\end{itemize}

When the above two conditions are satisfied, then the rate functional has the following form. Let $\Psi$
be the real function of two variables given by
\begin{equation}\label{defpsi}
\Psi(a,b)=\left\{
\begin{array}{ll}
a\log\frac ab+b -a & a,b>0\\
b & a=0, b\geq 0 \\
+\infty & a>0, b=0\,, \\
+\infty & \min\{a,b\}<0\,.
\end{array}
\right.
\end{equation}
The rate functional is given by
\begin{align}\label{zr-rate-gen}
I(\mu, Q)&=\sum_{\eta\in \Lambda_N}\sum_{(\eta,\eta')\in \tilde E_N}\Psi\Big(Q(\eta,\eta'),\mu(\eta)r(\eta,\eta')\Big)\nonumber \\
&=\sum_{\eta\in \Lambda_N}\sum_{(x,y)\in E_N}\Psi\Big(Q(\eta,\eta^{x,y}),\mu(\eta)r(\eta,\eta^{x,y})\Big)\,,
\end{align}
where the rates of transition $r$ are defined by
\begin{equation}\label{defr}
r(\eta,\eta^{x,y}):=\left\{
\begin{array}{ll}
g(\eta(x))c(x,y) & \textrm{if}\  x,y\in \Omega_N\,, (x,y)\in E_N\,,\\
\lambda_xc(x,y) & \textrm{if}\  x\in \partial \Omega_N, y\in \Omega_N\,, (x,y)\in E_N\,,\\
g(\eta(x))c(x,y) & \textrm{if}\  y\in \partial \Omega_N, x\in \Omega_N\,, (x,y)\in E_N\,.
\end{array}
\right.
\end{equation}
We use here the notation introduced at the end of Section \ref{empmf} and recall that our transition graph can be a multigraph with respect to transitions involving creation and annihilation of particles.
Note also that we employ the convention that $Q(\eta,\eta'):=0$ when $\eta'\not\in \Lambda_N$ (i.e.\  when $\eta'$ is such that $\eta'(x)<0$ for some $x\in \Omega_N$); all the terms in formula \eqref{zr-rate-gen} in which there appears a $\eta'\not\in \Lambda_N$ are of the type $\Psi\Big(Q(\eta,\eta'),0\Big)$ so that they are identically zero since $\Psi(0,0)=0$.

As discussed in \cite{BFG,BFG2}, the summability condition on the rates is important in order that the rate functional has the form \eqref{zr-rate-gen}; see, for example, Theorem~2.10 of~\cite{BFG}.  Furthermore, as outlined later on, summability is important in order to apply the contraction principle.  Note also that when there exists a product invariant measure of the form \eqref{prod-zr}, it is possible to show by a direct computation that the typical flow $\mu[\phi](\eta)r(\eta,\eta')$ is summable.
\begin{remark}\label{fipoi}
An important observation is that $\Psi(\cdot,\alpha)$ is the large deviation rate functional for a Poisson process of parameter $\alpha$. In other words, if $N_T^\alpha$ is a Poisson process of parameter $\alpha$ then we have
\begin{equation}
\mathbb P\left(\frac{N_T^\alpha}{T}\sim q\right)\simeq e^{-T\Psi(q,\alpha)}\,.
\end{equation}
Considering two independent Poisson processes of parameters $\lambda_1$ and $\lambda_2$ then by a direct application of the contraction principle and an explicit minimization we have
\begin{equation}
\mathbb P\Big(\frac{N_T^{\lambda_1}-N_T^{\lambda_2}}{T}\sim j\Big)\simeq e^{-T\Gamma_{\lambda_1,\lambda_2}(j)}\,,
\end{equation}
where
\begin{align}\label{defgamma}
\Gamma_{\lambda_1,\lambda_2}(j)&:=\inf_{\{q^\pm:q^+-q^-=j\}}\left\{\Psi(q^+,\lambda_1)+\Psi(q^-,\lambda_2)\right\}\nonumber \\
&=j\log\left(\frac{j}{2\sqrt{\lambda_1\lambda_2}}+\sqrt{\frac{j^2}{4\lambda_1\lambda_2}+1}\right)-j\log\sqrt{\frac{\lambda_1}{\lambda_2}}-\sqrt{j^2 +4\lambda_1\lambda_2}+\lambda_1+\lambda_2\,.
\end{align}
This class of functions will play an important role in the following and we define them also in the case when one of the $\lambda_i$ is zero. In particular we set
\begin{equation}
\left\{
\begin{array}{l}
\Gamma_{\alpha,0}(q):=\Psi(q,\alpha)\,, \\
\Gamma_{0,\alpha}(q):=\Psi(-q,\alpha) \,.
\end{array}
\right.
\end{equation}
\end{remark}
The functions \eqref{defgamma} are, by construction, convex in $j$ for fixed $\lambda_i$.  By a direct, but long, computation it is not difficult to show that they are also jointly convex in $\lambda_1,\lambda_2\geq 0$ for any fixed $j$.

\begin{remark}\label{remring}
We point out that the rate functional~\eqref{zr-rate-gen} is not just valid for zero-range models of particles, with rates as in \eqref{defr}, but indeed holds generically, under suitable assumptions, for the empirical measure and flow on a continuous-time Markov chain $\eta(t)$ with transition graph $(\Lambda_N, \tilde E_N)$ and transition rate $r(\eta,\eta')$ for jumping from $\eta\in \Lambda_N$ to $\eta'\in \Lambda_N$, when $(\eta,\eta')\in \tilde E_N$.
\end{remark}

\subsection{Current fluctuations}\label{seccur}
We introduce here the main problem studied in this paper.
Consider again $(\eta_t)_{t\in \mathbb R^+}$ a sample path of the zero-range model. The level 2.5 LDP gives the rate functional for the flow on the configuration space but we are interested rather in the current of particles flowing through the physical space. In our case, the particles move in the discrete physical space which is the directed graph $\left(\Omega_N\cup\partial\Omega_N, E_N\right)$ but, for notational convenience, it is useful to introduce also the corresponding unoriented graph $\left(\Omega_N\cup\partial\Omega_N, \mathcal E_N\right)$ where $\{x,y\}\in \mathcal E_N$ when at least one of $(x,y)$ and $(y,x)$ belongs to $E_N$.

As a function of the trajectory we define a discrete vector field $j_T$ on our unoriented graph $\left(\Omega_N\cup\partial\Omega_N, \mathcal E_N\right)$
as follows. Recall that a discrete vector field is a real function such that $j_T(x,y)=-j_T(y,x)$; this must hold for any $\{x,y\}\in \mathcal E_N$. We construct the collection of trajectory-dependent numbers  $(j_T(x,y))_{\{x,y\}\in  \mathcal E_N}$ as
\begin{equation}\label{defempc}
j_T(x,y):= \frac{\textrm{Net\ number\ of\ particles\ jumped\ from} \ x \ \textrm{to}\ y\ \textrm{in}\ [0,T]}{T}\,, \qquad (x,y) \in  E_N\,,
\end{equation}
defining $j_T(x,y)$ by anti-symmetry in the case when $\{x,y\}\in \mathcal E_N$ but $(x,y)\not \in E_N$.
Once again we remark that we are defining the current also on edges $\{x,y\}$ with one of $x$ and $y$ belonging to $\Omega_N$ and the other to $\partial \Omega_N$; this can be done since our transition graph is a multigraph.

The main problem we are interested in is a large deviation result for $j_T$ in the limit when $T\to +\infty$. We need therefore to establish a result that can be informally stated as
\begin{equation}\label{palloncini2}
\mathbb P_\eta\left(j_T\sim j\right)\simeq e^{-T\mathcal R^N(j)}\,,
\end{equation}
where we write simply $j$ for the value of the discrete vector field $j=(j(x,y))_{\{x,y\}\in \mathcal E_N}$.

\smallskip

Several comments are in order here. In the above formula we again do not really need to specify the initial condition since, under our assumption on the function $g$,
the long-time behaviour will be independent of the initial condition. Our result here is obtained for fixed $N$ in the limit of large $T$; the rate functional depends on $N$ as indicated by the superscript. Later, in Section~\ref{secdiscon}, we will consider the convergence of $\mathcal R^N$ when $N\to +\infty$ and relate the limiting behavior with the computation of current fluctuations from the macroscopic fluctuation theory \cite{MFT}.

As we will prove, the rate functional $\mathcal R^N$ is equal to $+\infty$ unless the discrete vector field $j$ is divergence-free, i.e.\
\begin{equation}
\nabla\cdot j(x):=\sum_{y: \{x,y\}\in \mathcal E_N}j(x,y)=0\,, \qquad x\in \Omega_N\,. \label{divfree}
\end{equation}
In particular, in one dimension a divergence-free discrete vector field is necessarily constant and $\mathcal R^N$ thus reduces to a real function of a single variable which is the current across each edge.

\subsection{Contraction principle}\label{sec:contraction}

Our proposal is to deduce the large deviation principle \eqref{palloncini2} from \eqref{palloncini} using the fact that the empirical current on the real lattice $\left(\Omega_N\cup\partial \Omega_N, E_N\right)$ can be computed from the empirical flow on the transition graph $(\Lambda_N, \tilde E_N)$ and we can therefore apply the contraction principle.
Starting from a flow $Q$ on the transition graph $(\Lambda_N,\tilde E_N)$ we can define
a flow $q$ on the real-space lattice $\left(\Omega_N\cup \partial \Omega_N, E_N \right)$.
Using the notation introduced in Section~\ref{sec2point5}, the flow
$q$ is defined by
\begin{equation}\label{ancora-arpa-d}
q(x,y)=\sum_\eta Q(\eta,\eta^{x,y})\,.
\end{equation}

By an obvious direct computation, when we apply formula \eqref{ancora-arpa-d} to the empirical flow $Q_T$ then the corresponding flow $q_T$ is related to the empirical current by
\begin{equation}\label{def-corrente-fisica}
j_T(x,y)= q_T(x,y)- q_T(y,x), \qquad \{x,y\}\in \mathcal E_N\,,
\end{equation}
where we define $q_T(x,y):=0$ when $(x,y)\not \in E_N$.

The divergence of a flow $q$ on the lattice $\left(\Omega_N\cup \partial\Omega_N, E_N \right)$ is again defined analogously to \eqref{zero-div-zr} as
\begin{equation}
\nabla\cdot  q(x):= \sum_{y\,:\, (x,y)\in E_N} q(x,y)-\sum_{y\,:\, (y,x)\in E_N} q(y,x)\,.
\end{equation}
With these definitions we have the following result.
\begin{lemma}
Consider a summable flow $Q$ on $(\Lambda_N,\tilde E_N)$ and let $q$ be the flow defined by \eqref{ancora-arpa-d} on the physical lattice $\left(\Omega_N\cup \partial\Omega_N, E_N \right)$.  If $\nabla\cdot Q(\eta)=0$ for any $\eta\in \Lambda_N$, then also $\nabla\cdot q(x)=0$ for any $x\in \Omega_N$.
\end{lemma}
\begin{proof}
By Lemma 4.1 in \cite{BFG} we have that any divergence-free summable flow can be decomposed as a superposition $Q=\sum_{C} a_C Q_C$  where the $Q_C$ are flows associated to elementary cycles and $a_C$ are constant coefficients. More precisely,
let $C:=\left(\eta_0,\eta_1,\dots , \eta_{k-1},\eta_k\right)$ be an elementary cycle
on the transition graph $(\Lambda_N,\tilde E_N)$. This means a collection of distinct vertices of $\Lambda_N$
such that $(\eta_i,\eta_{i+1})\in \tilde E_N$ and $\eta_0=\eta_k$.
The associated flow $Q_C$ is defined by
\begin{equation}\label{QC}
Q_C(\eta,\eta'):=\left\{
\begin{array}{ll}
1 & \textrm{if}\ (\eta,\eta')\in C\,\\
0 & \textrm{otherwise}\,.
\end{array}
\right.
\end{equation}
By linearity it is enough to prove that $\nabla\cdot  q_C(x)=0$, $x\in \Omega_N$, where the flow $q_C$ on $(\Omega_N\cup \partial\Omega_N, E_N)$ is
obtained from $Q_C$ by \eqref{ancora-arpa-d}. Indeed the outgoing flux $\sum_{y:(x,y)\in E_N} q_C(x,y)$ from $x\in \Omega_N$ coincides with the number of edges of the type $(\eta,\eta^{x,y})$ in the cycle $C$. Likewise
the incoming flux $\sum_{y:(y,x)\in E_N} q_C(y,x)$ in $x\in \Omega_N$ coincides with the number of edges of the type $(\eta^{x,y},\eta)$ in the cycle $C$. We have that the incoming flux and the outgoing flux at $x$ are equal since going around the cycle $C$ the number
of particles $\left(\eta_i(x)\right)_{i=0,\dots,k}$ present at site $x$ starts with the value $\eta_0(x)$ and returns to the same value $\eta_k(x)=\eta_0(x)$. This ends the proof.
\end{proof}

With the divergence-free property in hand, the contraction principle gives the following representation for the rate functional $\mathcal R^N$
\begin{equation}\label{contrpr}
\mathcal R^N(j)=\inf_{\left\{(\mu,Q): j=j[Q]\right\}} I(\mu,Q)\,,
\end{equation}
where $j[Q]$ is the discrete vector field such that
$j[Q](x,y)=q(x,y)-q(y,x)$, where $q$ is obtained from $Q$ by \eqref{ancora-arpa-d}.

\smallskip

We emphasize again that, since the transition graph is infinite, the validity of the contraction principle is not straightforward. This is one of the points where the superlinear growth of the function $g$ is relevant.  The discussion after formula \eqref{palloncini} gives a hint of the issues involved; in particular, with the strong $L^1_+(\tilde E_N)$ topology on flows, the map \eqref{ancora-arpa-d} is continuous.

\section{The variational argument}

Applying the contraction principle \eqref{contrpr} we obtain the following theorem which is one of our main results.
\begin{theorem}\label{Th1}
The empirical current \eqref{defempc} satisfies a large deviation principle \eqref{palloncini2} on $\mathbb R^{\mathcal E}$ when $T\to \infty$ with rate functional
\begin{equation}\label{laff}
\mathcal R^N(j)=\left\{
\begin{array}{ll}
\sum_{(x,y)\in \mathcal E_N}\Gamma_{\phi^*_xc(x,y),\phi^*_yc(y,x)}(j(x,y))\,,& \textrm{if}\ \div j(x)=0\  \forall x\in \Omega_N\\
+\infty & \textrm{if}\ \div j\neq 0\,,
\end{array}
\right.
\end{equation}
where the functions $\Gamma$ are as defined in Remark \ref{fipoi} and $\phi^*$ is the unique positive solution of
\begin{equation}\label{stazd}
2\phi^*_x\sum_{y:\{x,y\}\in \mathcal E_N}c(x,y)=\sum_{y:\{x,y\}\in \mathcal E_N}\sqrt{j(x,y)^2+4\phi^*_x\phi^*_yc(x,y)c(y,x)}\,, \qquad x\in \Omega_N\,.
\end{equation}
\end{theorem}
\begin{proof}
Given a function $\phi:\Omega_N\to \mathbb R^+$, recall that we denote by $\mu[\phi]$ the product measure \eqref{prod-zr}. Given a discrete vector field $F:E_N\to \mathbb R$, we define modified rates
\begin{equation}
r^F(\eta,\eta^{x,y}):=r(\eta,\eta^{x,y})e^{F(x,y)}\,,
\end{equation}
in terms of the original rates $r$ given in \eqref{defr}.

Consider a divergence-free discrete vector field $(j(x,y))_{(x,y)\in \mathcal E_N}$ and an arbitrary function
$(\phi_x)_{x\in \Omega_N \cup \partial \Omega_N}$ such that $\phi_x=\lambda_x$ for $x\in \partial \Omega_N$. Then there exists
a unique discrete vector field $(F(x,y))_{\{x,y\}\in \mathcal E_N}$ such that
\begin{equation}\label{uniqueF}
\phi_x c(x,y)e^{F(x,y)}-\phi_yc(y,x)e^{F(y,x)}=j(x,y)\,, \qquad \{x,y\}\in \mathcal E_N\,,
\end{equation}
and it is the logarithm of the unique positive solution of the second-degree equation
$X^2\phi_x c(x,y)-Xj(x,y)-\phi_yc(y,x)=0$, i.e.
\begin{equation}\label{Fsol}
F(x,y)=\log \frac{j(x,y)+\sqrt{j(x,y)^2+4\phi_x\phi_yc(x,y)c(y,x)}}{2\phi_x c(x,y)}\,, \qquad \{x,y\}\in \mathcal E_N\,.
\end{equation}
Note that we have $F(x,y)=-F(y,x)$ in \eqref{Fsol}, as required.

We restrict our variational computation to measures of the form $\mu[\phi]$ for arbitrary functions $\phi$ and to flows on the transition graph $(\Lambda_N,\tilde E_N)$ of the form
\begin{equation}\label{min-flu-zr}
Q[\phi,F](\eta,\eta'):=\mu[\phi](\eta)r^F(\eta,\eta')\,.
\end{equation}
The condition that the flow \eqref{min-flu-zr} is divergence-free corresponds to the condition that $\mu[\phi]$ is invariant for the zero-range dynamics with rates $r^F$. Since in \eqref{uniqueF} $j$ is divergence-free, taking the divergence of the left-hand side we obtain the stationary condition \eqref{inv-gen}. This means that any flow $Q[\phi,F]$ with $F$ defined by \eqref{Fsol} is divergence-free.

In Appendix \ref{appB} we will show that we can restrict to this class of measures and flows to obtain the  global minimizer in \eqref{contrpr}.

By the relation $\sum_\eta \mu[\phi](\eta)g(\eta(x))=\phi_x$, which straightforwardly follows from~\eqref{prod-zr}, we obtain that for a given divergence-free discrete vector field $j$ together with an arbitrary function $\phi$, there exists a unique discrete vector field $F$ that yields $j=j[Q[\phi,F]]$.  We consider $j$ as fixed and denote by $Q[\phi]$ the flow \eqref{min-flu-zr} for such $F$, i.e.\ the $F$ of~\eqref{Fsol}.

By a direct computation we have
\begin{equation}\label{energy-fij}
I(\mu[\phi],Q[\phi])=\sum_{(x,y)\in \mathcal{E}_N}\Gamma_{\phi_xc(x,y),\phi_yc(y,x)}(j(x,y))\,.
\end{equation}
We now have to minimize over the function $\phi$ and there are no constraints apart from the positivity. Condition \eqref{stazd} corresponds to setting the gradient equal to zero; uniqueness of the positive solution follows by the convexity in $\phi$ of \eqref{energy-fij}.
\end{proof}

An important fact which should be stressed here is that \eqref{laff} does not depend on the form of the function $g$ that determines the model and, in particular, the final result is the same as we would get for a model of independent particles. This phenomenon also arises in the macroscopic approach with the corresponding variational problem obtained by the macroscopic fluctuation theory \cite{MFT}. In the case that there are no boundary sources, this reduction to the independent case does not work and the rate function depends on $g$, see for example \cite{CGP}.

From a mathematical viewpoint, the fact that when there are no sources the rate functional loses its universal form and becomes $g$-dependent can be understood from the failure of the computations in Appendix~\ref{appB}. Indeed without sources, the total number of particles is fixed from the outset and, due to this, we no longer have simple formulas like \eqref{ultimaeq}. The minimization is more constrained (and thus harder) which results in different phenomenology. Intuitively, this comes from the fact that in the conserved-particle case, the system can organize into effective current-carrying structures, like traveling waves, which can persist for long times; in the boundary-driven case such configurations are destroyed after a finite time (since particles leave the system) and are therefore not relevant for long-time behavior.

\section{The one-dimensional case}

In this section, we single out the one-dimensional case where all computations can be done explicitly and reveal a simpler structure of the rate functional that we will extend to the general case afterwards. In this section $\Omega_N=\left\{1,\dots, N\right\}$, $\partial \Omega_N=\{0,N+1\}$ and $E_N$ connect nearest-neighbor sites; we denote as $\lambda_L$ and $\lambda_R$ the chemical potentials associated to the left and right boundary nodes.
In particular, we show here the form of the rate functional $\mathcal R^N$  using two methods: a heuristic coupling argument as in \cite{BDGJL1} and an exact variational computation by the contraction principle applied to
\eqref{laff}. Both approaches can be extended to general graphs. We point out that the expression of $\mathcal R^N$ in the one-dimensional case had also been previously obtained in \cite{HRS}
using an ansatz and combinatorial arguments.

In Sections~\ref{coupling} to~\ref{ident} below we restrict ourselves to the symmetric homogeneous case before considering inhomogeneous one-dimensional models in Section~\ref{zrgengen}.

\subsection{The coupling argument}\label{coupling}
The rate functional $\mathcal R^N$ for the symmetric homogeneous case, with unit weights $c$, has been explicitly
computed in \cite{BDGJL1} with a coupling argument. The exact statement in \cite{BDGJL1} is that the rate function of the LDP for
$j_T$ coincides with the rate functional of
\begin{equation}\label{equivP}
\frac 1T\left(N_T^{\frac{\lambda_L}{N+1}}-N_T^{\frac{\lambda_R}{N+1}}\right)\,,
\end{equation}
where $N_T^{\frac{\lambda_L}{N+1}}$ and $N_T^{\frac{\lambda_R}{N+1}}$ are
two independent Poisson processes with respective parameters
$\frac{\lambda_L}{N+1}$ and $\frac{\lambda_R}{N+1}$.
According to \eqref{defgamma} this is given by
\begin{equation}
\mathcal R^N(j) =\Gamma_{\frac{\lambda_L}{N+1},\frac{\lambda_R}{N+1}}(j)\,.
\label{ratecurrzero}
\end{equation}
We are now going to summarize the heuristic argument in \cite{BDGJL1}.

\smallskip

In the long-time limit the net flow of particles across each bond of the one-dimensional lattice can be computed as the number of particles created at the left boundary that leave the system through the right boundary minus the number of particles created at the right boundary that leave the system thorough the left boundary. This is because we are in the case when condensation does not happen and the probability to have an accumulation of particles on any site of the lattice is superexponentially small.

Particles are created at the left boundary with a Poisson process of parameter $\lambda_L$ and at the right boundary with a Poisson process of parameter $\lambda_R$. Disregarding holding times at each vertex, the trajectory of each particle is that of a simple random walk independent from the trajectories of all the other particles. Indeed the interaction among the particles modifies only the holding times and not the probability of jumping to the left or to the right.

The probability that one particle created at the left boundary will leave the system through the right boundary is therefore the probability that a simple random walk starting at $1\in \Omega_N\cup\{0,N+1\}$
will reach $N+1$ before $0$, and this is given by $\frac{1}{N+1}$. If we mark independently with probability $\frac{1}{N+1}$ each event of a Poisson process of parameter $\lambda_L$ we obtain a Poisson process of parameter $\frac{\lambda_L}{N+1}$. This is the approximate law of the number of particles created at the left boundary that exit from the right boundary, in the long-time limit. Likewise a Poisson process of parameter $\frac{\lambda_R}{N+1}$ is the approximate law of the number of particles created at the right boundary that exit from the left boundary, in the long time. This argument implies that the large deviations for $j_T$ are the same as those of the process \eqref{equivP}.

\subsection{The variational argument}\label{variational}

In the homogeneous one-dimensional case, again with all weights $c$ set to unity, we have that \eqref{energy-fij} is given by
\begin{align}\label{ilrate-zr}
I(\mu[\phi],Q[\phi])=\sum_{x=0}^N\Gamma_{\phi_x,\phi_{x+1}}(j)\,,
\end{align}
and the minimization conditions \eqref{stazd} become
\begin{equation}\label{staz-fi}
4\phi_x=\sqrt{j^2+4\phi_{x-1}\phi_{x}}+\sqrt{j^2+4\phi_x\phi_{x+1}}\,, \qquad x=1,\dots,N\,.
\end{equation}
Recall that in \eqref{staz-fi} we have $\phi_0=\lambda_L$ and $\phi_{N+1}=\lambda_R$. We now need to find the unique solution to the coupled nonlinear algebraic equations \eqref{staz-fi}.

\begin{lemma}\label{quadratico}
The unique positive solution of \eqref{staz-fi} with the boundary conditions $\phi_0=\lambda_L$ and $\phi_{N+1}=\lambda_R$ is given by
\begin{equation}\label{sol-zr}
\phi^*_x=A+Bx+Cx^2
\end{equation}
with
\begin{equation}\label{coeffa}
\left\{
\begin{array}{l}
A=\lambda_L\,, \\
B= \frac{-2\lambda_L+\sqrt{(N+1)^2j^2+4\lambda_L\lambda_R}}{(N+1)}\,,\\
C=\frac{\lambda_L+\lambda_R-\sqrt{(N+1)^2j^2+4\lambda_L\lambda_R}}{(N+1)^2}\,.
\end{array}
\right.
\end{equation}
\end{lemma}
The proof of this lemma is in Appendix \ref{apquadratico}.

\smallskip

This is then the unique critical point of \eqref{ilrate-zr} and
it corresponds to its global minimum. Note that, as in the corresponding macroscopic problem in the framework of hydrodynamic limits \cite{BDGJL1}, the minimizer \eqref{coeffa} depends only on the square of the current. This and other relations will be discussed later on.

We obtained therefore that the rate functional $\mathcal R^N(j)$ is given by \eqref{ilrate-zr} evaluated at the critical point $\{\phi^*_x\}_{x=1}^N$ which is the form \eqref{sol-zr} with coefficients determined by \eqref{coeffa}. This representation of the rate functional should be compared with the representation \eqref{ratecurrzero} obtained by the coupling approach. The identification of the two formulas is not trivial but results from the following computations.

\subsection{The identification}\label{ident}
It will be useful to recall the definition of the inverse hyperbolic function
$$
\asinh(x)=\log\left(x+\sqrt{x^2+1}\right)
$$
and its additive property
\begin{equation}\label{add-h}
\asinh (u) \pm \asinh (v)=\asinh\left(u\sqrt{1+v^2}\pm v\sqrt{1+u^2}\right)\,.
\end{equation}
Using this we can write \eqref{defgamma} as
\begin{equation}\label{basic}
\Gamma_{\lambda_1,\lambda_2}(j):=j\asinh \left(\frac{j}{2\sqrt{\lambda_1\lambda_2}}\right)-j\log\sqrt{\frac{\lambda_1}{\lambda_2}}-\sqrt{j^2 +4\lambda_1\lambda_2}+\lambda_1+\lambda_2\,.
\end{equation}

The identification between the two different expressions of the rate functional follows by the following additivity formula.
\begin{lemma}\label{lemma-add}\textbf{[Discrete additivity lemma]}
For any $j$ and $a+b>0$ we have
\begin{equation}\label{unitstep}
\inf_{\phi\geq 0}\left\{\Gamma_{\lambda,a\phi}(j)+\Gamma_{b\phi,\mu}(j)\right\}=\Gamma_{\frac{b\lambda}{a+b},\frac{a\mu}{a+b}}(j)\,.
\end{equation}
\end{lemma}
\begin{proof}
We consider separately the first term, the second term and the remaining terms in \eqref{basic}.
The critical point associated to the minimization problem in \eqref{unitstep} is characterized by a generalization of condition \eqref{staz-fi} given by
\begin{equation}\label{staz-fi-gen}
2(a+b)\phi=\sqrt{j^2+4\lambda a \phi}+\sqrt{j^2+4b\phi\mu}\,.
\end{equation}
For the critical value of $\phi$, characterized by \eqref{staz-fi-gen},  we have using the additivity formula \eqref{add-h} that
\begin{equation}\label{add-h-special}
j\asinh \left(\frac{j}{2\sqrt{\lambda a\phi}}\right) + j\asinh \left(\frac{j}{2\sqrt{b\phi\mu}}\right)=j\asinh\left(\frac{j}{2\sqrt{\frac{b\lambda a\mu}{(a+b)^2}}}\right).
\end{equation}
Hence the first term in~\eqref{basic} matches in the equality. For the second term we have
$$
-j\log\sqrt{\frac{\lambda}{a\phi}}-j\log\sqrt{\frac{b\phi}{\mu}}=-j\log\sqrt{\frac{b\lambda}{a\mu}}
$$
which is also compatible with the statement. Finally, for the remaining terms we have the identity
\begin{multline}
-\sqrt{j^2 +4\lambda a\phi}+\lambda+a\phi-\sqrt{j^2 +4\mu b\phi}+\mu+b\phi  \\
=-\sqrt{j^2 +4\frac{ba\lambda\mu}{(a+b)^2}}+\frac{b\lambda}{a+b}+\frac{a\mu}{a+b}\,,
\end{multline}
which it is possible to show under the critical condition \eqref{staz-fi-gen} with some tricky computations.

\end{proof}
Using the above lemma we deduce the following main result for the homogeneous one-dimensional zero-range process.
\begin{theorem}
For symmetric homogeneous one-dimensional nearest-neighbour zero-range dynamics, we have that
\begin{equation}\label{quallali}
\mathcal R^N(j)=\inf_{\left\{\phi: \phi_0=\lambda_L,\phi_{N+1}=\lambda_R\right\}}\sum_{x=0}^N\Gamma_{\phi_x,\phi_{x+1}}(j)=\Gamma_{\frac{\lambda_L}{N+1},\frac{\lambda_R}{N+1}}(j)\,,
\end{equation}
where the unique minimizer in \eqref{quallali} is characterized in Lemma \ref{quadratico}.
\end{theorem}
\begin{proof}
The proof is obtained by induction on $N$. The case with $N=1$ is given by a special case of Lemma
\ref{lemma-add}. Let us now suppose that the coincidence of the two formulas is true for $N$ and show that it is then true also for $N+1$.   By the inductive hypothesis we have
\begin{align}
\inf_{\left\{\phi : \phi_0=\lambda_L, \phi_{N+2}=\lambda_R\right\}}&\left\{\sum_{x=0}^{N+1}\Gamma_{\phi_x.\phi_{x+1}}(j)\right\} \notag\\
&=\inf_{\phi_{N+1}}\left\{\inf_{\left\{\phi_0=\lambda_L,\phi_1,\dots, \phi_{N}\right\}}\left\{\sum_{x=0}^{N}\Gamma_{\phi_x.\phi_{x+1}}(j)\right\}+\Gamma_{\phi_{N+1},\lambda_R}(j)\right\}\\
&=\inf_{\phi_{N+1}} \left\{\Gamma_{\frac{\lambda_L}{N+1},\frac{\phi_{N+1}}{N+1}}(j)+\Gamma_{\phi_{N+1},\lambda_R}(j)\right\}\\
&=\Gamma_{\frac{\lambda_L}{N+2},\frac{\lambda_R}{N+2}}(j)\,,
\end{align}
where in the last step we used Lemma \ref{lemma-add} with $a=\frac{1}{N+1}$
and $b=1$ for which we get $\frac{a}{a+b}=\frac{1}{N+2}$ and $\frac{b}{a+b}=\frac{N+1}{N+2}$.
\end{proof}

\subsection{One-dimensional inhomogeneous zero-range model} \label{zrgengen}

We now extend the analysis to spatially inhomogeneous one-dimensional zero-range dynamics with left and right hopping rates from site $x$ given by $c(x,x-1)$ and $c(x,x+1)$ with the condition that $\inf_x c(x,x\pm 1)>0$ so that particles can cross the system. The computations are similar to the ones of Section \ref{variational}; we do not discuss all the details here  since the inhomogeneous model can, in fact, be included in the general treatment of Section~\ref{gensec}.  For the present case we have
\begin{equation}\label{rateinomogeneo}
\mathcal R^N(j)=\inf_{\left\{\phi: \phi_0=\lambda_L, \phi_{N+1}=\lambda_R\right\}}\sum_{x=0}^N \Gamma_{\phi_xc(x,x+1),\phi_{x+1}c(x+1,x)}(j)\,.
\end{equation}

To elucidate the structure of~\eqref{rateinomogeneo} and understand how it can be written in closed form, we first introduce some more notation.

\smallskip

Consider the continuous-time random walk $X(t)$ on the lattice $\Omega_N\cup \partial \Omega_N$ with jumps from $x$ to nearest-neighbor $y$ at rate $c(x,y)$. Given $S\subseteq \Omega_N\cup \partial \Omega_N$ we define
\begin{equation}
\tau_S:=\inf \left\{t : X(t)\in S\right\}\,,
\end{equation}
and
\begin{equation}
\tau^+_S:=\inf \left\{t >\tau^1\,, X(t)\in S\right\}\,,
\end{equation}
where $\tau^1$ denotes the time of the first jump out of the starting state.
Given $y \in \Omega_N\cup \partial \Omega_N$ and disjoint $A,B\subseteq \Omega_N\cup \partial \Omega_N$, we write
$p_y(A,B)$ for the probability that a random walker starting at $y$ at time zero reaches $A$ before $B$, i.e.\
\begin{equation}
p_y(A,B):=\mathbb P_y\left(X(\tau_{A\cup B})\in A\right)\,,
\end{equation}
and consequently we have $p_y(B,A)=1-p_y(A,B)$ under irreducibility assumptions.

The following result will be useful.
\begin{lemma}
	Consider the random walk with rates $(c(x,y))_{(x,y)\in  E_N}$ on $\Omega_{N+1}\cup\partial \Omega_{N+1}=\{1,\dots, N+1\}\cup\{0,N+2\}$. For any $k+1 \in \Omega_{N+1}$ we have
	\begin{align}\label{useful}
	p_1(N+2,0)&=\frac{p_{1}(k+1,0)c(k+1,k+2)p_{k+2}(N+2,k+1)}{c(k+1,k+2)p_{k+2}(N+2,k+1)+c(k+1,k)p_k(0,k+1)}\,, \\
	p_{N+1}(0,N+2)&=\frac{ p_{N+1}(k+1,N+2)c(k+1,k)p_k(0,k+1)}{c(k+1,k+2)p_{k+2}(N+2,k+1)+c(k+1,k)p_k(0,k+1)}\,.
	\end{align}
\end{lemma}
\begin{proof}
	These results emerge directly from a statistical physics treatment of first-passage/exit probabilities.   We prove here the first one and note that the proof of the second one is essentially the same. Consider the Markov chain $X(t)$ with transition rates given by the $c$ and define the two stopping times $\tau_{\{0,N+2\}}$ and
	$\tau_{\{0,k+1\}}$. We have that $p_1(N+2,0)=\mathbb P_1(X(\tau_{\{0,N+2\}})=N+2)$ and this can be expressed as
	$\mathbb P_1(X(\tau_{\{0,N+2\}})=N+2)=\mathbb P_1(X(\tau_{\{0,N+2\}})=N+2|X(\tau_{\{0,k+1\}})=k+1)\mathbb P_1(X(\tau_{\{0,k+1\}})=k+1)$. Now $\mathbb P_1(X(\tau_{\{0,k+1\}})=k+1)=p_{1}(k+1,0)$ and, by the strong Markov property, $\mathbb P_1(X(\tau_{\{0,N+2\}})=N+2|X(\tau_{\{0,k+1\}})=k+1)=\mathbb P_{k+1}(X(\tau_{\{0,N+2\}})=N+2)$. To determine this last probability, we note that when the Markov chain starts with $X(0)=k+1$ we have
	\begin{multline}\label{erse}
	\mathbb P_{k+1}\left(X(\tau^+_{\{0,k+1,N+2\}})=a\right)\\=\left\{
	\begin{array}{ll}
	\frac{c(k+1,k+2)p_{k+2}(N+2,k+1)}{c(k+1,k+2)+c(k+1,k)}\,, & \textrm{if}\ a=N+2\\
	\frac{c(k+1,k)p_{k}(0,k+1)}{c(k+1,k+2)+c(k+1,k)}\,, & \textrm{if}\ a=0\\
	1-\frac{c(k+1,k+2)p_{k+2}(N+2,k+1)+c(k+1,k)p_{k}(0,k+1)}{c(k+1,k+2)+c(k+1,k)} & \textrm{if}\ a=k+1\,.
	\end{array}
	\right.
	\end{multline}
	Now consider a sequence of independent identically distributed (i.i.d.) random variables that assume the values $N+2$ with probability $p_R$, $0$ with probability $p_L$, and $k$+1 with probability $1-p_R-p_L$; the probability to observe $N+2$ in this sequence before observing $0$ is given by $\frac{p_R}{p_R+p_L}$. 
	To be concrete, we can associate a random variable with the distribution~\eqref{erse} to each of the i.i.d.\ cycles obtained by cutting the trajectory of a Markov chain starting at $k+1$ at the return times to state $k+1$; the exit probability $\mathbb P_{k+1}(X(\tau_{\{0,N+2\}})=N+2)$ is indeed just the probability of seeing $N+2$ before $0$ in this sequence of random variables. 
	Hence, using the expressions in \eqref{erse}, we obtain the first relation in \eqref{useful}.
\end{proof}
We are now in a position to generalize the result in Section \ref{variational} as follows.
\begin{lemma}
	We have that formula \eqref{rateinomogeneo} coincides with
	\begin{equation}
	\mathcal R^N(j)=\Gamma_{\lambda_L c(0,1)p_1(N+1,0), \lambda_R c(N+1,N) p_N(0,N+1)}(j)\,.
	\end{equation}
\end{lemma}
\begin{proof}
	The proof is again obtained by induction. Consider first the case $N=1$. By Lemma \ref{lemma-add} we have that
	$$\mathcal R^1(j)=\Gamma_{\lambda_L \frac{c(0,1)c(1,2)}{c(1,2)+c(1,0)}, \lambda_R \frac{c(2,1)c(1,0)}{c(1,2)+c(1,0)}}(j)$$ which coincides with the statement of the lemma for $N=1$ since,
	by a direct computation, it is easy to see that $p_1(2,0)=\frac{c(1,2)}{c(1,2)+c(1,0)}$ and
	$p_1(0,2)=\frac{c(1,0)}{c(1,2)+c(1,0)}$.
	
	Let us now suppose that the statement is true for any $M\leq N$ and show that it is then true also for $N+1$.
	By the inductive hypothesis we can write
	\begin{align}
	\mathcal R^{N+1}(j)=&\inf_{\phi_{k+1}}\left\{\Gamma_{\lambda_L c(0,1)p_1(k+1,0), \phi_{k+1} c(k+1,k) p_k(0,k+1)}(j)\right. \notag \\
	&\left.+\Gamma_{\phi_{k+1} c(k+1,k+2)p_{k+2}(N+2,k+1), \lambda_R c(N+2,N+1) p_{N+1}(k+1,N+2)}(j)\right\}\,,
	\end{align}
	and, using Lemma \ref{lemma-add}, we obtain that the right-hand side above is equal to
	\begin{equation}
	\Gamma_{\frac{\lambda_L c(0,1)p_{1}(k+1,0)c(k+1,k+2)p_{k+2}(N+2,k+1)}{c(k+1,k+2)p_{k+2}(N+2,k+1)+c(k+1,k)p_k(0,k+1)}, \frac{\lambda_R c(N+2,N+1) p_{N+1}(k+1,N+2)c(k+1,k)p_k(0,k+1)}{c(k+1,k+2)p_{k+2}(N+2,k+1)+c(k+1,k)p_k(0,k+1)}}(j)\,.
	\end{equation}
	The proof is then concluded using the relations \eqref{useful}.
\end{proof}

\smallskip

We note that analogous results for current fluctuations in the one-dimensional inhomogeneous case were presented in Section 6.3 of~\cite{RH} with a heuristic discussion of effective input and output rates for specific sites, in particular the recursions in~\cite{RH} are related to the computations in this section.   The present work puts these results on a more rigorous footing and clarifies the relation to exit probabilities.  We emphasize again that there is no dependence on the specific form of the function $g$ (assuming it is superlinear) and the fluctuations coincide with those of free particles.  The fluctuations should even be the same for some non-Markovian zero-range dynamics (e.g., non-exponential waiting times at each site but constant probabilities to jump left and right) in regimes where there is no condensation, see~\cite{CMH} for some related discussion.

\section{Scaling limits}\label{sec:scaling}

In this section we determine the scaling limit of the functionals that we computed in the previous section. We restrict ourselves for simplicity to the symmetric homogeneous cases. The natural mathematical treatment would be in terms of Gamma convergence \cite{braides} but a full discussion would make this part of the paper rather technical and long. Instead we once again use a direct approach, explaining the scaling limit informally, and stressing the correspondence between the microscopic formulation and the macroscopic one.
Here the macroscopic formulation is that of macroscopic fluctuation theory \cite{MFT} which uses the large deviation rate functionals from the diffusive hydrodynamic scaling limit.

\subsection{Macroscopic fluctuation theory}

According to the MFT \cite{MFT}, the joint density and current large deviations from the diffusive hydrodynamic scaling limit of particle systems on a time interval $[0,T]$ and on a bounded domain $\Omega\subseteq \mathbb R^d$ are determined by the joint rate functional
\begin{equation}\label{principal}
\frac 14\int_0^Tdt\int_\Omega dx \frac{\Big|j(x,t)+D(\rho(x,t))\nabla \rho(x,t)\Big|^2}{\sigma(\rho(x,t))}\,.
\end{equation}
Here $j(x,t)\in \mathbb R^d$ is a smooth vector field and $\rho(x,t)\in \mathbb R^+$; they represent the observed space-time macroscopic current and density which have to be related by the continuity equation $\partial_t\rho+\nabla\cdot j=0$. The coefficients $D$ and $\sigma$ are respectively the diffusion coefficient and the mobility.

When we are interested in the fluctuations of the current through the system we have to suitably minimize the functional \eqref{principal}. In particular, to compute the large deviation rate functional for observing an average current $j(x)$ over a long time, we need to minimize \eqref{principal} under the constraint $\frac 1T\int_0^Tj(s,x)ds=j(x)$, and then take the limit $T\to +\infty$.

By general arguments \cite{MFT} we have that the result of this procedure for the rate functional is $+\infty$ unless $\nabla \cdot j=0$. For a class of models, viz.\ those that do not exhibit dynamical phase transitions, we have that the minimizer in \eqref{principal} is obtained by time-independent paths.  With our choice of superlinear $g$, the boundary-driven zero-range model indeed has no dynamical phase transitions and the rate functional is therefore given by
\begin{equation}\label{Umacro}
U(j)= \frac 14 \inf_\rho \int_\Omega dx \frac{\Big|j(x)+D(\rho(x))\nabla \rho(x)\Big|^2}{\sigma(\rho(x))}\,,
\end{equation}
where the infimum is over density profiles satisfying some boundary conditions imposed at $\partial \Omega$ \textrm{and suitable regularity conditions}.

In the case of homogeneous zero-range dynamics on a lattice with $c(x,y)=1$ for any pair of nearest-neighbor sites, we have that $D(\rho)=\Phi'(\rho)$ and $\sigma(\rho)=\Phi(\rho)$ where the function $\Phi(\rho)$ is the inverse of the function $\rho(\Phi)$ defined by $\rho(\Phi)=\Phi \frac{d}{d \Phi}\log Z(\Phi)$, with the partition function $Z(\Phi)$ as introduced below~\eqref{prod-zr}.
In the case of a strong interaction with the boundary sources (i.e.\ when the boundary rates scale with $N$ in the same way as the bulk rates) then the infimum in \eqref{Umacro} is over all the density profiles that satisfy the boundary conditions $\rho(x)=\rho(\lambda_x)$, $x\in \partial \Omega$ where $\lambda_x$ are the rates appearing in the generator.

\smallskip

In short, for the symmetric zero-range model in strong contact with the sources, the rate functional of
formula \eqref{Umacro} becomes
\begin{equation}\label{MFT-formula}
  \frac 14 \inf_\rho \int_{\Omega} \frac{\left|j(x)+\Phi'(\rho{(x)})\nabla \rho{(x)}\right|^2}{\Phi(\rho(x))}dx\,,
\end{equation}
and the infimum is over all density profiles such that $\Phi(\rho(x))=\lambda_x$ for $x\in \partial\Omega$.
If we define $\varphi(x):=\Phi(\rho(x))$, we see that
\eqref{MFT-formula} can be written as
\begin{equation}\label{MFT-formula2}
  \frac 14 \inf_\varphi\int_{\Omega}\frac{\left|j(x)+\nabla \varphi(x)\right|^2}{\varphi(x)} dx\,,
  \end{equation}
where the infimum is over the $C^1$ functions $\varphi$ such that $\varphi(x)=\lambda_x$ for $x\in \partial \Omega$.   Note that in~\eqref{MFT-formula}, $\Phi$ encodes a dependence on the function $g$ appearing in the rates, but in \eqref{MFT-formula2} there is no longer any such dependence. This change of variable shows that the fluctuations of the current are model-independent and coincide with those of the free-particle case. This is exactly the same phenomenon observed in formula \eqref{laff} and the minimization in \eqref{MFT-formula2} is the continuous counterpart of that on the right-hand side of \eqref{laff}.

In the one-dimensional case, the current must be independent of $x$; we set $\Omega=[0,1]$ and the infimum is over $\varphi$ such that $\varphi(0)=\lambda_L$
and $\varphi(1)=\lambda_R$.
By the convexity for $\varphi \geq 0$ of the functional in \eqref{MFT-formula2} we have a unique minimizer. The infimum is characterized by the Euler-Lagrange equations giving, with the boundary conditions, the differential problem
\begin{equation}\label{euler}
\left\{
\begin{array}{l}
-2\varphi''\varphi+(\varphi')^2=j^2\,,\\
\varphi(0)=\lambda_L\,, \  \varphi(1)=\lambda_R\,,
\end{array}
\right.
\end{equation}
which can be solved \cite{BDGJL1}.
We search for a quadratic solution $\varphi(x)=A+Bx+Cx^2$ and we obtain that, neglecting the boundary conditions, any quadratic solution solves the equation in \eqref{euler} provided that $B^2-4AC=j^2$. Note that this is exactly the same condition as to have a quadratic solution to the discrete equation \eqref{staz-fi}, as explained in Lemma~\ref{quadratico}.
Imposing the boundary conditions, we obtain
\begin{equation}\label{coeffb}
\left\{
\begin{array}{l}
A=\lambda_L\,, \\
B= -2\lambda_L+\sqrt{j^2+4\lambda_L\lambda_R}\,,\\
C=\lambda_L+\lambda_R-\sqrt{j^2+4\lambda_L\lambda_R}\,,
\end{array}
\right.
\end{equation}
which is the continuous counterpart of \eqref{coeffa}. Here again there is another possible choice of sign in front of the square root but the sign that we select is the one which gives a non-negative $\varphi$ in $[0,1]$.
Let us then use $\varphi^*$ for the quadratic optimal solution of \eqref{euler}. By a direct computation we have
\begin{equation}\label{U}
U(j)=\frac 14 \int_0^1\frac{\left(j+(\varphi^*)'(u)\right)^2}{\varphi^*(u)} du=\Gamma_{\lambda_L,\lambda_R}(j)\,.
\end{equation}
Perhaps surprisingly, we thus find that the continuous computation results in the same function $\Gamma$ as in the discrete case (with particular values of the parameters).

\subsection{From discrete to continuous}
\label{secdiscon}

We now show that we recover the continuous result just described  by the scaling limit of the discrete computations of the previous sections and, moreover, we see that at the discrete level we have a structure which is a direct counterpart of the continuous one.

First we recall some simple statements which will be used in our analysis and
can be obtained by elementary computations.
We have that
\begin{equation}\label{primo-limite}
\lim_{x\to +\infty}\left(\sqrt{j^2+x^2}-x\right)=0\,,
\end{equation}
and, furthermore, if $x_N$ and $y_N$ are two positive sequences such that
\begin{equation}
\left\{
\begin{array}{l}
\lim_{N\to +\infty}\frac{N}{y_N}=\kappa^{-1}\,,\\
\lim_{N\to +\infty}\left(x_N-y_N\right)=j\,,
\end{array}
\right.
\label{ancora-uno}
\end{equation}
then we have
\begin{equation}\label{secondo-ingrediente}
\lim_{N\to +\infty}N\Psi\left(x_N,y_N\right)=\frac{j^2}{2\kappa}\,,
\end{equation}
with the function $\Psi$ defined as in~\eqref{defpsi}. For a positive real value $a$ we have
\begin{equation}\label{terzo-ingr}
\log\left(\frac{x+\sqrt{x^2+a^2}}{a}\right)=\frac xa +o(x) \quad \text{as $x \to 0$}\,.
\end{equation}

Using the above facts, an elementary computation gives the following.  Consider $j, a\in \mathbb R$, $\lambda\in \mathbb R^+$ and $N$ an integer number; we have
\begin{equation}\label{unlimite}
\lim_{N\to +\infty}N^2\Gamma_{\lambda,\lambda +\frac aN}\left(\frac jN\right)=\frac 14 \frac{(j+a)^2}{\lambda}\,.
\end{equation}

To determine the scaling limit for the discrete rate functional and understand correctly the scaling factors that we need to consider, we start the discussion with a symmetry argument.
Consider two Markov chains having the same transition graph, one with transition rates $\left\{r(x,y)\right\}_{(x,y)\in E}$ and one with transition rates
$\left\{\lambda r(x,y)\right\}_{(x,y)\in E}$ where $\lambda$ is a positive parameter. Let us denote by respectively
$I$ and $I_\lambda$ the corresponding joint rate functional of the LDP for
the empirical measure and flow.  With an explicit expression (i.e.\  \eqref{zr-rate-gen} in the zero-range case)
we get the relationship  $I_\lambda(\mu, Q)=\lambda I\left(\mu, \frac Q\lambda\right)$.
From this we immediately obtain the relation
\begin{equation}
\mathcal R_{\lambda}^N(j)=\lambda \mathcal R^N\left(\frac{j}{\lambda}\right)\,,
\label{simm-zr}
\end{equation}
where $\mathcal R^N$ is the rate functional for the current of the zero-range model (as before) and $\mathcal R_{\lambda}^N$ is the rate functional when the rates are multiplied by a factor of $\lambda$.

We focus first on the one-dimensional case, noting that here we simply have $|\Omega_N|=N$. To consider our symmetric zero-range model in the diffusive scaling limit the jump rates
must be multiplied by a factor of $N^2$. Moreover, in this diffusive scaling limit we assign a mass equal to $N^{-1}$ to each particle so that, to a given macroscopic current $j$, there corresponds a number of jumps (microscopic current) equal to $jN$. We therefore need to consider the rate functional
$\mathcal R^N_{N^2}\left(jN\right)$ for large $N$. Since $N$ is the size of the system for large $N$, the microscopic rate functional becomes proportional to $N$ and consequently the non-trivial scaling limit is
\begin{equation}\label{illimite}
\lim_{N\to +\infty}\frac{1}{N}\mathcal R^N_{N^2}(jN)=\lim_{N\to+\infty}N\mathcal R^N\left(\frac{j}{N}\right)\,.
\end{equation}
The limit in \eqref{illimite} can be directly computed  using formula \eqref{ratecurrzero} and we reobtain  the macroscopic result \eqref{U}
\begin{equation}
\lim_{N\to+\infty}N\mathcal R^N\left(\frac{j}{N}\right)=\Gamma_{\lambda_L,\lambda_R}(j)=U(j)\,.
\end{equation}
We took a pointwise limit here but, building up a suitable topological framework, a Gamma convergence could be established too \cite{braides}.

We now turn to the $d$-dimensional case, considering $\Omega_N$ as the portion of the regular squared lattice contained on a bounded domain $\Omega\subseteq \mathbb R^d$. In this case we have $|\Omega_N|\sim N^d$. We consider the scaling limit of the functional $\sum_{\{x,y\}\in \mathcal E_N}\Gamma_{\phi_x,\phi_y}(j(x,y))$; the mass of each particle is $1/N^d$ and the volume is proportional to $N^d$. We give a hint of the argument here. Consider a function $\phi$ that is obtained as $\phi_x=\varphi(x/N)$ for a $C^1([0,1])$ positive function $\varphi$. Let $j(x)$ be a smooth continuous vector field and let its flow across $\Sigma_{(x,y)}$, the $(d-1)$-dimensional plaquette dual to the edge $(x,y)$, be $j_N(x,y):=\int_{\Sigma_{(x,y)}}j\cdot n\, d\sigma$. Here $n$ is the unit vector oriented in the direction of $(x,y)$ and $d\sigma$ is the differential surface element. Note that the area of $\Sigma_{(x,y)}$ is $1/N^{d-1}$ and $j_N(x,y)$ is therefore of the same order. Using the result \eqref{unlimite}, we obtain as a generalized Riemann sum
\begin{align}
 \lim_{N\to +\infty} &\frac{1}{N^d}\sum_{\{x,y\}\in \mathcal E_N}\Gamma_{N^2\varphi\left(\frac{x}{N}\right),N^2\varphi\left(\frac{y}{N}\right)}\left(j_N(x,y)N^d\right) \notag \\
&=\lim_{N\to +\infty} N^{2-d}\sum_{\{x,y\}\in \mathcal E_N}\Gamma_{\varphi\left(\frac{x}{N}\right),\varphi\left(\frac{y}{N}\right)}\left(j_N(x,y)N^{d-2}\right)\\
&=\frac 14\int_\Omega dx \frac{\left|j(x)+\nabla\varphi(x)\right|^2}{\varphi(x)}\,.
\end{align}
The argument $j_N(x,y)N^d$ in the first line is due to the fact that the microscopic current on the graph should correspond to the flow multiplied by $N^d$ since the mass of each particle is $N^{-d}$. We can apply \eqref{unlimite} to deduce the last equality since $j_N(x,y)N^{d-2}\sim j_i(x)/N$ when $y=x+e_i/N$. 

\subsection{Additivity principle}

The additivity principle, introduced in \cite{bd}, is a general principle that allows one to deduce, for a class of one-dimensional diffusive systems, the validity of rate functionals of the form \eqref{MFT-formula}, \eqref{MFT-formula2}.
The additivity principle fails when there is a dynamical phase transition \cite{MFT}.

To be specific, let us consider again a one-dimensional diffusive system with diffusion coefficient $D(\rho)$ and mobility $\sigma(\rho)$ on an interval of length $L$ and in contact with density reservoirs at left and right with values $\rho_{L}$ and $\rho_R$ respectively.  The additivity principle predicts that large deviations for the current over long times have rate functional given by \cite{bd}
\begin{equation}\label{additivity}
{U}^L_{\rho_L,\rho_R}(j)=\inf_\rho \frac 14\int_0^L \frac{\left(j+D(\rho(x))\nabla \rho(x)\right)^2}{\sigma(\rho(x))}dx\,,
\end{equation}
where the infimum is over all density profiles that match $\rho(0)=\rho_L$ and $\rho(L)=\rho_R$; we here explicitly indicate both the system length and the boundary densities in our notation for $U$. Splitting the interval $(0,L)$ in \eqref{additivity} into two intervals
$(0,\ell_1)$ and $(\ell_1,L)$ of length $\ell_1$ and $\ell_2=L-\ell_1$, we have
\begin{equation}\label{addpro}
U^L_{\rho_L,\rho_R}(j)=\inf_{\rho_\ell}\left\{U^{\ell_1}_{\rho_L,\rho_\ell}(j)+U^{\ell_2}_{\rho_\ell,\rho_R}(j)\right\}\,.
\end{equation}
The additivity property \eqref{addpro} is satisfied by all solutions of the variational problem \eqref{additivity}, see for example \cite{Cap,bd}.

As we have already seen, the surprising feature of the zero-range model is that the rate functional of the continuous variational problem, when written in terms of the boundary chemical potentials,
$U^L_{\lambda_L,\lambda_R}(j)$ coincides exactly with a rate functional obtained in a microscopic discrete computation; we have
 \begin{equation}
U^L_{\lambda_L,\lambda_R}(j) = \Gamma_{\frac{\lambda_L}{L}, \frac{\lambda_R}{L}}(j)
 \end{equation}
 and~ the additivity formula~\eqref{addpro} then reduces directly to~\eqref{unitstep} with the identifications $a=\frac{1}{\ell_1}$, $b=\frac{1}{\ell_2}$, $\lambda=\frac{\lambda_L}{\ell_1}$ and $\mu=\frac{\lambda_R}{\ell_2}$.

\section{General graphs and electrical networks}\label{gensec}

The above computation can be interpreted in terms of an analogous electrical network construction. We here explain this construction and apply it to the case of a generic graph, demonstrating how to use a network reduction technique with special rules for components connected in series and in parallel, as well as a generalized star-triangle correspondence.  Note that the electrical components that we consider can be associated to \emph{non-reversible dynamics}; in other words, we present a network reduction theory for non-reversible random walks which coincides with the classic electrical network approach in the case of symmetric components.  A theory of electrical networks for non-reversible Markovian dynamics has been developed recently in \cite{B} with the introduction of auxiliary electrical components. In our construction we use just one type of electrical component with each individual component characterized by a left and a right conductivity.  This geometric approach is related to the \emph{trace process} \cite{L} used in the computation of the so-called \emph{capacities}.

\subsection{Non-reversible electric-like components}

To be concrete, the basic non-reversible electric-like component of our construction is an unoriented edge, i.e.\ an element of $\mathcal E_N$, that is characterized by a right conductivity $c_R$, with a corresponding right resistance $(c_R)^{-1}$, and a left conductivity $c_L$, with corresponding left resistance $(c_L)^{-1}$. There is an \emph{energy} or \emph{cost} associated to the passage of a current $j\in \mathbb R$ through this component when it is connected to a potential $\lambda_L$ at the left boundary and a potential $\lambda_R$ at the right boundary. This cost function is given by
\begin{equation}\label{defenergy}
\mathcal U_{\lambda_L,\lambda_R}^{c_L,c_R}(j):=\Gamma_{\lambda_Lc_L,\lambda_Rc_R}(j)\,.
\end{equation}
The typical current flowing through the system from left to right is the one that minimizes the cost (with the minimum value being zero); it is given by $\lambda_Lc_L-\lambda_Rc_R$ which can, of course, be positive or negative. Note that, since the component is not reversible and has a left conductivity which may be different from the right one, passage of current is possible even in the case that the two extremes are connected to the same potential. The classic theory of electrical networks is recovered for components with $c_L=c_R$.

\smallskip

The general framework is then the following.  We consider our arbitrary network with internal nodes $\Omega_N$ and boundary nodes $\partial \Omega_N$ and we fix the applied potential to the values $(\lambda_v)_{v\in \partial \Omega_N}$. Each link of the network is characterized by a left and a right conductivity; as before, $\mathcal E_N$ is the set of undirected edges while $E_N$ is the set of directed edges. If $\{x,y\}\in \mathcal E_N$ and $(x,y),(y,x)\in E_N$ we call  $c(x,y)$ the conductivity from $x$ to $y$ and $c(y,x)$ the conductivity in the opposite direction. When $(x,y)\not\in E_N$ we set $c(x,y)=0$. The cost or energy dissipated observing a current $(j(x,y))_{(x,y)\in \mathcal E_N}$,  with $j(x,y)=-j(y,x)$, and a potential configuration $(\phi_x)_{x\in \Omega_N\cup \partial \Omega_N}$ such that $(\phi_x)_{x\in \partial \Omega_N}=(\lambda_x)_{x\in \partial\Omega_N}$ is given by
\begin{eqnarray}
\sum_{(x,y)\in \mathcal E_N}\mathcal U^{c(x,y),c(y,x)}_{\phi_x,\phi_y}(j(x,y))\,,
\end{eqnarray}
which is exactly the functional \eqref{energy-fij} related to the relative entropy of two zero-range processes. Recall that $\mathcal U^{c(x,y),c(y,x)}_{\phi_x,\phi_y}(j(x,y))=\mathcal U^{c(y,x),c(x,y)}_{\phi_y,\phi_x}(-j(x,y))$. Charge is obviously conserved such that its creation or destruction occurs only at the boundary sites and we must always have $\nabla\cdot j(x)=0$ for any $x\in \Omega_N$, i.e.\ a divergence-free field as defined in~\eqref{divfree}.

When only the current $(j(x,y))_{(x,y)\in \mathcal{E}_N}$ and the boundary potentials $(\phi_x)_{x\in \partial \Omega_N}=(\lambda_x)_{x\in \partial \Omega_N}$ are fixed, the cost is given by $\mathcal R^N(j)$ in \eqref{laff}. When instead the potential $(\phi_x)_{x\in\Omega_N}$ is fixed then the cost is obtained by minimizing over all the possible currents satisfying the zero-divergence constraint.

We next show that, for some specific problems, we can reduce the complexity of the network by substituting the original electrical components with effective ones according to special rules for components connected in series, in parallel, or in a generalized star configuration.  In each case, the costs associated to the original graph and to the reduced one are the same. Note that some analogous results for ``effective renormalized hopping rates'' in the series case are implied by the heuristic discussion in~\cite{RH} but the general electrical-component framework presented here, in particular the star-triangle transformation, provides a formal basis for the treatment of more complex networks.

\subsubsection{Components in series}
\begin{figure}
	
	\centering

	\includegraphics[width=0.7\textwidth]{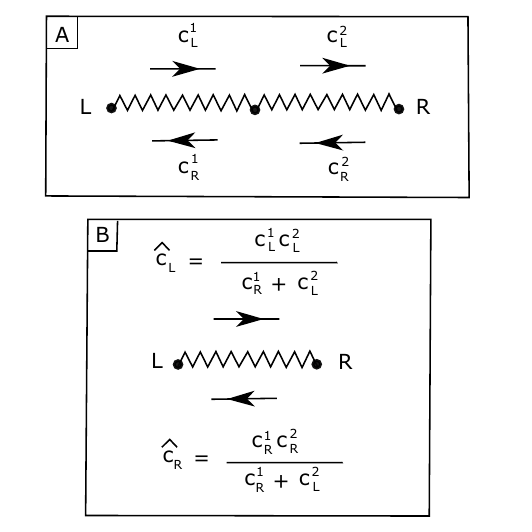}
	\caption{Two electrical components in series connecting the nodes $L$ and $R$ (diagram A). The equivalent single electrical component with effective conductivities (diagram B).}\label{comp_series}
\end{figure}
Consider two electrical components connected in series with component $1$ on the left of component $2$ as in Figure~\ref{comp_series}.  The vertex set
of the graph is here given by just three  points; dropping the $N$ subscript and labelling the intermediate node as $M$ we have $\Omega=\{M\}$ with boundary nodes $\partial \Omega = \{L,R\}$.  We denote by $c^i_L$ ($c^i_R$) the left (right) conductivity of the $i^\mathrm{th}$ component and we suppose that $c^1_R+ c^2_L>0$. From the random walk perspective, this latter condition is rather natural since when it is violated the walker is trapped in the middle point and consequently charge accumulates there (see also the similar discussion at the start of Section~\ref{zrgengen}).

We now put the two extremes in contact with fixed potentials $\lambda_L$ and $\lambda_R$ at $L$ and $R$, respectively. Trivially, any current with zero  divergence on $\{M\}$ satisfies  $j(L,M)=j$
and $j(M,R)=j$ for some $j\in \mathbb{R}$.
If the current $j$ is fixed, then the optimum potential $\phi_M$ at the junction point $M$ between components $1$ and $2$ is obtained by minimizing the sum of the energies; using Lemma~\ref{lemma-add} we thus have
\begin{align}
  \inf_{\phi_M}\left[\mathcal U_{\lambda_L,\phi_M}^{c^1_L,c^1_R}(j)+\mathcal U_{\phi_M,\lambda_R}^{c^2_L,c^2_R}(j)\right] &= \inf_{\phi_M}\left[\Gamma_{\lambda_Lc^1_L,\phi_Mc^1_R}(j)+\Gamma_{\phi_Mc^2_L,\lambda_Rc^2_R}(j)\right]\nonumber \\
  &=\Gamma_{\lambda_L\frac{c^1_Lc^2_L}{c^1_R+c^2_L},\lambda_R\frac{c^2_Rc^1_R}{c^1_R+c^2_L}}(j) \nonumber \\
  &=\mathcal U_{\lambda_L,\lambda_R}^{\frac{c^1_Lc^2_L}{c^1_R+c^2_L},\frac{c^2_Rc^1_R}{c^1_R+c^2_L} }(j)\,.
\end{align}
Let us next imagine that the above electrical components are part of a larger network, but still with an intermediate node $M$ shared only by these two components in series. As long as the value of the potential $\phi_M$ is not fixed and not relevant, then the node $M$ can be removed and  the two electrical components in series can be substituted by one single electrical component having effective left and right conductivities $\hat c_L$, $\hat c_R$ given respectively by
\begin{equation}
\hat c_L=\frac{c^1_Lc^2_L}{c^1_R+c^2_L}\,, \qquad \hat c_R= \frac{c^2_Rc^1_R}{c^1_R+c^2_L}\,.
\end{equation}
Note that the above expressions are not symmetric under the exchange of $1$ and $2$ but we recover the classic rules of electrical networks when $c^i_L=c^i_R$.

\subsubsection{Components in parallel} Turning now to investigate electrical components in parallel, we begin with a simple observation:
\begin{lemma}\label{testossimoro} For each $ j \in \mathbb{R}$ we have
	\begin{equation}\label{ossimoro}
	\inf_{\{ (j_1,j_2): j_1+j_2= j\}}\big\{\Gamma_{\lambda_1,\lambda_2}(j_1)+\Gamma_{\mu_1,\mu_2}(j_2)\}=
	\Gamma_{\lambda_1+\mu_1,\lambda_2+\mu_2}(j) \,.
	\end{equation}
\end{lemma}
\begin{proof}

\begin{figure}
	
	\centering
	
	\includegraphics[width=0.7\textwidth]{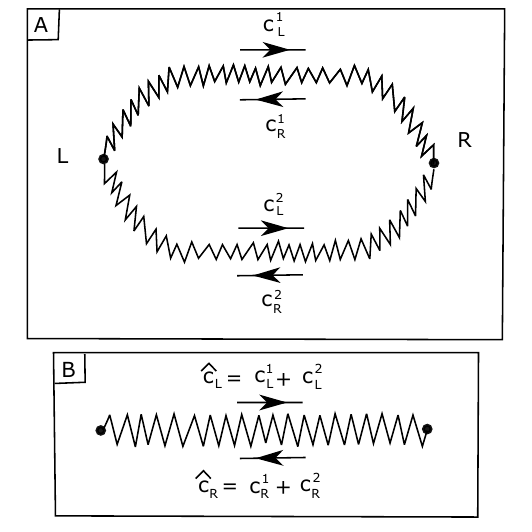}
	\caption{Two electrical components in parallel connecting the nodes $L$ and $R$ (diagram A). The equivalent single electrical component with effective conductivities (diagram B).}\label{parallelo}
\end{figure}

	The proof of this lemma can be obtained simply by a probabilistic argument.  Let $N^{\lambda_1}_T$, $N^{\lambda_2}_T$, $N^{\mu_1}_T$, $N^{\mu_2}_T$ be independent Poisson processes. The joint rate functional for the pair $\frac{1}{T}\left(N^{\lambda_1}_T-N^{\lambda_2}_T,N^{\mu_1}_T-N^{\mu_2}_T\right)$ is given by
	$(j_1,j_2)\mapsto \Gamma_{\lambda_1,\lambda_2}(j_1)+\Gamma_{\mu_1,\mu_2}(j_2)$. By the contraction principle,
	the rate functional for the sum of the two components of the pair is then given by the left-hand side of \eqref{ossimoro}.  However, we also have that the sum of the components has the same distribution as
	$\frac{1}{T}\left(N^{\lambda_1+\mu_1}_T-N^{\lambda_2+\mu_2}_T\right)$, whose rate functional is the right-hand side of~\eqref{ossimoro}.
\end{proof}

The above lemma implies the following rule of equivalence for two electrical components connected in parallel.
Consider again component $1$ with left (right) conductivity $c^1_L$ ($c^1_R$) and component $2$ with left (right) conductivity $c^2_L$ ($c^2_R$). We now arrange the two components so that their left and right extremes are coincident and set to potentials $\lambda_L$ and $\lambda_R$, respectively (see Figure \ref{parallelo}).
The cost when current $j_1$ crosses component $1$ and current $j_2$ crosses component $2$ is given by
\begin{equation}\label{enotte}
\mathcal U_{\lambda_L, \lambda_R}^{c^1_L,c^1_R}(j_1)+
\mathcal U_{\lambda_L, \lambda_R}^{c^2_L,c^2_R}(j_2)\,.
\end{equation}
When the potentials $\lambda_L$, $\lambda_R$ are fixed and we are not interested in the exact currents flowing across each of the two components but only in the total current flowing between the two nodes, then the cost associated to a total current $j$ is obtained minimizing \eqref{enotte} under the constraint $j_1+j_2=j$.
Indeed, by Lemma \ref{testossimoro}, the minimum of \eqref{enotte} over all $j_1$ and $j_2$ such that
$j_1+j_2=j$ is given by $\mathcal U_{\lambda_L, \lambda_R}^{c^1_L+c^2_L,c^1_R+c^2_R}(j)$.
This means that the pair of electrical components can be substituted by one single component having effective conductivities given by
\begin{equation}
\hat c_L=c^1_L+c^2_L\,, \qquad \hat c_R=c^1_R+c^2_R\,.
\end{equation}

\subsubsection{Star-triangle relation}

The star-triangle transformation allows us to replace a star graph with three rays by a triangle as in Figure \ref{star-triangle}. While in the classic theory of electrical networks there are transformations in both directions, here we will have a transformation from the star to the triangle but generally not the reverse. Note that going from the star to the triangle decreases by one the number of vertices and thus reduces the complexity of the network.  We label the vertices of the star $A$, $B$, $C$, and $O$; only the boundary nodes $A$, $B$, and $C$ are retained in the triangle.

To be concrete, the equivalence holds and we can replace the star by the triangle as long as we do not observe the potential in the central node $O$ of the star and simply impose that the total current entering/exiting from each node of the triangle through the edges of the triangle is identical to that in the star with the same potentials on the boundary nodes (see Lemma \ref{lemma-a3} for a precise statement).

\begin{figure}
	
	\centering
	
	\includegraphics[width=0.7\textwidth]{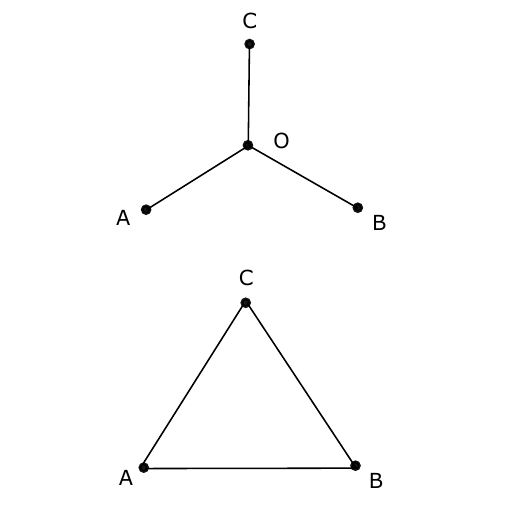}
	\caption{Geometry of the star-triangle relation.  Non-reversible components on the links of the star at the top are to be replaced by effective components on the links of the triangle at the bottom.} \label{star-triangle}
\end{figure}

\smallskip

We denote the hopping rates (conductivities) in the star by $c(A,O)$, $c(B,O)$, $c(C,O)$, $c(O,A)$, $c(O,B)$, $c(O,C)$ and those in the triangle by $\hat{c}(A,B)$, $\hat{c}(B,C)$, $\hat{c}(C,A)$, $\hat{c}(B,A)$, $\hat{c}(C,B)$, $\hat{c}(A,C)$.
The algebraic relations between conductivities which have to be satisfied in order that a star-triangle relation may hold are not that easily guessed but some clues come from considering the special cases when one among $j(A,O)$, $j(B,O)$, and $j(C,O)$ is identically zero.
It turns out that when the conductivities $c$ are fixed, the conductivities $\hat c$ are uniquely determined by
\begin{equation}\label{star-tria}
\left\{
\begin{array}{l}
\hat{c}(A,B)=\frac{c(A,O)c(O,B)}{c(O,A)+c(O,B)+c(O,C)}\,,\\
\hat{c}(B,A)=\frac{c(B,O)c(O,A)}{c(O,A)+c(O,B)+c(O,C)}\,,\\
\hat{c}(B,C)=\frac{c(B,O)c(O,C)}{c(O,A)+c(O,B)+c(O,C)}\,,\\
\hat{c}(C,B)=\frac{c(C,O)c(O,B)}{c(O,A)+c(O,B)+c(O,C)}\,,\\
\hat{c}(C,A)=\frac{c(C,O)c(O,A)}{c(O,A)+c(O,B)+c(O,C)}\,,\\
\hat{c}(A,C)=\frac{c(A,O)c(O,C)}{c(O,A)+c(O,B)+c(O,C)}\,.\\
\end{array}
\right.
\end{equation}
Further intuition and a formal justification of these transformations is given later in this section with a natural interpretation in terms of the \emph{trace process} in Section~\ref{trace}.

\smallskip

We first discuss the inversion of~\eqref{star-tria}.  To show that this is not possible in general we observe that, if we construct the product of the conductivities $\hat c$ obtained in the previous formulas,  we obtain the same result going clockwise or anticlockwise in the triangle, i.e.\ we have
\begin{equation}\label{ungnu}
\hat{c}(A,B)\hat{c}(B,C)\hat{c}(C,A)=\hat{c}(A,C)\hat{c}(C,B)\hat{c}(B,A)\,.
\end{equation}
The reverse procedure is therefore possible only when \eqref{ungnu} is satisfied.  Moreover, even in this case, the conductivities $c$ are not uniquely defined as we now show.

We first state two simple facts that will be useful in the following.
Consider the system of equations in the variables $x$, $y$, $z$ given by
\begin{equation}\label{primo-sist}
\left\{
\begin{array}{l}
x=a_1y\\
y=a_2z\\
z=a_3x
\end{array}
\right.
\end{equation}
where $a_1$, $a_2$, $a_3$ are real numbers such that $a_1a_2a_3=1$. Equations \eqref{primo-sist} have a one-parameter family of solutions given by
\begin{equation}\label{sol-primo}
x=\left(\frac{a_1}{a_3}\right)^{\frac 13}\ell\,,\ \ \  y=\left(\frac {a_2}{a_1}\right)^{\frac 13}\ell\,,\ \ \ z=\left(\frac {a_3}{a_2}\right)^{\frac 13}\ell\,, \ \ \ \ell\in \mathbb R\,.
\end{equation}
Consider also the system of equations in the variables $x$, $y$, $z$ given by
\begin{equation}\label{secondo-sist}
\left\{
\begin{array}{l}
b_1=\frac{xy}{x+y+z}\\
b_2=\frac{yz}{x+y+z}\\
b_3=\frac{zx}{x+y+z}\\
\end{array}
\right.
\end{equation}
with $b_1$, $b_2$, $b_3$ real numbers.
Equations~\eqref{secondo-sist} have a unique solution given by
\begin{equation}\label{sol-seconda}
x=\frac{b_1b_3+b_1b_2+b_2b_3}{b_2}\,, \ \ \ y=\frac{b_1b_3+b_1b_2+b_2b_3}{b_3}\,, \ \ \ z=\frac{b_1b_3+b_1b_2+b_2b_3}{b_1}\,.
\end{equation}

Using these facts we now consider inverting \eqref{star-tria}. Let us introduce the notation
$\hat{\gamma}(X,Y):=\frac{\hat c(X,Y)}{\hat c(Y,X)}$ where $X$ and $Y$ are any two of $A$, $B$ and $C$. Likewise we define ${\gamma}(X,O):=\frac{c(X,O)}{c(O,X)}$ where $X$ is $A$, $B$ or $C$.
Using~\eqref{star-tria}, we obtain the equations
\begin{equation}
\left\{
\begin{array}{l}
{\gamma}(A,O)=\hat{\gamma}(A,B){\gamma}(B,O)\\
{\gamma}(B,O)=\hat{\gamma}(B,C){\gamma}(C,O)\\
{\gamma}(C,O)=\hat{\gamma}(C,A){\gamma}(A,O)\\
\end{array}
\right.
\end{equation}
which, since $\hat{\gamma}(A,B)\hat{\gamma}(B,C)\hat{\gamma}(C,A)=1$ due to \eqref{ungnu}, have a one-parameter family of solutions
\begin{equation}\label{sol-primo2}
{\gamma}(A,O)=\left(\frac{\hat{\gamma}(A,B)}{\hat{\gamma}(C,A)}\right)^{\frac 13}\ell\,,\ \ \   {\gamma}(B,O)=\left(\frac{\hat{\gamma}(B,C)}{\hat{\gamma}(A,B)}\right)^{\frac 13}\ell\,,\ \ \ {\gamma}(C,O)=\left(\frac{\hat{\gamma}(C,A)}{\hat{\gamma}(B,C)}\right)^{\frac 13}\ell\,, 
\end{equation}
where $\ell\in \mathbb R^+$ as in~\eqref{sol-primo}.  From~\eqref{star-tria} we can also get equations relating the three variables $c(O,A)$, $c(O,B)$, $c(C,O)$ to ${\gamma}(A,O)$, ${\gamma}(B,O)$, ${\gamma}(C,O)$:
\begin{equation}
\left\{
\begin{array}{l}
\frac{\hat{c}(A,B)}{{\gamma}(A,O)}=\frac{c(O,A)c(O,B)}{c(O,A)+c(O,B)+c(O,C)}\\
\frac{\hat{c}(B,C)}{{\gamma}(B,O)}=\frac{c(O,B)c(O,C)}{c(O,A)+c(O,B)+c(O,C)}\\
\frac{\hat{c}(C,A)}{{\gamma}(C,O)}=\frac{c(O,A)c(O,C)}{c(O,A)+c(O,B)+c(O,C)}.
\end{array}
\right.
\end{equation}
By \eqref{sol-seconda} the solution to the above system of equations is
\begin{equation}
\left\{
\begin{array}{l}
c(O,A)=\frac 1\ell \frac{\frac{\hat{c}(A,B)\hat{c}(B,C)}{\left(\frac{\hat{\gamma}(B,C)}{\hat{\gamma}(C,A)}\right)^{\frac 13}}+\frac{\hat{c}(A,B)\hat{c}(C,A)}{\left(\frac{\hat{\gamma}(A,B)}{\hat{\gamma}(B,C)}\right)^{\frac 13}}+\frac{\hat{c}(C,A)\hat{c}(B,C)}{\left(\frac{\hat{\gamma}(C,A)}{\hat{\gamma}(A,B)}\right)^{\frac 13}}}{\frac{\hat{c}(B,C)}{\left(\frac{\hat{\gamma}(B,C)}{\hat{\gamma}(A,B)}\right)^{\frac 13}}}\\
c(O,B)=\frac 1\ell \frac{\frac{\hat{c}(A,B)\hat{c}(B,C)}{\left(\frac{\hat{\gamma}(B,C)}{\hat{\gamma}(C,A)}\right)^{\frac 13}}+\frac{\hat{c}(A,B)\hat{c}(C,A)}{\left(\frac{\hat{\gamma}(A,B)}{\hat{\gamma}(B,C)}\right)^{\frac 13}}+\frac{\hat{c}(C,A)\hat{c}(B,C)}{\left(\frac{\hat{\gamma}(C,A)}{\hat{\gamma}(A,B)}\right)^{\frac 13}}}{\frac{\hat{c}(C,A)}{\left(\frac{\hat{\gamma}(C,A)}{\hat{\gamma}(B,C)}\right)^{\frac 13}}}\\
c(O,C)=\frac 1\ell \frac{\frac{\hat{c}(A,B)\hat{c}(B,C)}{\left(\frac{\hat{\gamma}(B,C)}{\hat{\gamma}(C,A)}\right)^{\frac 13}}+\frac{\hat{c}(A,B)\hat{c}(C,A)}{\left(\frac{\hat{\gamma}(A,B)}{\hat{\gamma}(B,C)}\right)^{\frac 13}}+\frac{\hat{c}(C,A)\hat{c}(B,C)}{\left(\frac{\hat{\gamma}(C,A)}{\hat{\gamma}(A,B)}\right)^{\frac 13}}}{\frac{\hat{c}(A,B)}{\left(\frac{\hat{\gamma}(A,B)}{\hat{\gamma}(C,A)}\right)^{\frac 13}}}\,,
\end{array}
\right.
\end{equation}
and hence we see that the conductivities from the central node to the boundaries of the star are not uniquely defined but depend on the free parameter $\ell\in \mathbb R^+$.  Finally, we can obtain the values $c(A,O)=c(O,A){\gamma}(A,O)$, $c(B,O)=c(O,B){\gamma}(B,O)$
$c(C,O)=c(O,C){\gamma}(C,O)$ and see, in contrast, that the conductivities from the boundaries to the central node are fixed numbers and do not depend on the free parameter $\ell$.

\smallskip

Formally, the star-triangle transformation follows by the following statement.
\begin{lemma}\label{lemma-a3}
For any $\phi_A,\phi_B,\phi_C\geq 0$ and for any $j(A,O), j(B,O), j(C,O)$ such that $j(A,O)+j(B,O)+j(C,O)=0$, we have
\begin{align}\label{staruguale}
\inf_{\phi_O}&\left\{\Gamma_{\phi_Ac(A,O),\phi_O c(O,A)}(j(A,O)) +\Gamma_{\phi_Bc(B,O),\phi_Oc(O,B)}(j(B,O))\right.  \nonumber\\ & \left. +~ \Gamma_{\phi_Cc(C,O),\phi_Oc(O,C)}(j(C,O))\right\} \nonumber \\
&\quad =  \inf
\left\{\Gamma_{\phi_A\hat{c}(A,B),\phi_B\hat{c}(B,A)}(j(A,B))+\Gamma_{\phi_B\hat{c}(B,C),\phi_C\hat{c}(C,B)}(j(B,C)) \right. \nonumber \\ & \left.
\quad\quad+~\Gamma_{\phi_C\hat{c}(C,A),\phi_A\hat{c}(A,C)}(j(C,A))\right\}
\end{align}
where the $\inf$ on the right-hand side is over
\begin{multline}\label{forcontj}
\left\{j(A,C)+j(A,B)=j(A,O), j(C,A)+j(C,B)=j(C,O), \right. \\ \left. j(B,A)+j(B,C)=j(B,O)\right\}
\end{multline}
and where the conductivities involved are related by \eqref{star-tria} or equivalent formulation.
\end{lemma}
\begin{proof}
A proof by a direct computation is highly non-trivial; we present here a probabilistic proof. Since the rate functionals that we are interested in are the same for any superlinear zero-range process, we consider independent particles. We start from the star geometry and interpret it as a graph with one single node $\Omega_N=\{O\}$ in contact with three ghost sites $\partial \Omega_N=\{A,B,C\}$. By formula \eqref{laff} the large deviation rate functional for the current through the edges of this star is finite only on currents that satisfy the zero-divergence condition at $O$, which is one of the hypotheses of the lemma; in that case the finite value of the rate functional is given by the left-hand side of \eqref{staruguale}.

\smallskip

We now do a higher-level large deviation principle; rather than just observing the flow across the edges of the star, we label particles and for each of them we look at the boundary site where it is created and the boundary site where it exits.
By remark \ref{remarktilde} we can imagine that the particles that are created at $O$ from the different boundaries immediately jump outside the system, since the large deviation rate functionals for long times are invariant with respect to how long particles spend on the site. Note that a particle at $O$ will exit from site $X$, with $X=A,B,C$ with probability respectively $p_X:=\frac{c(O,X)}{c(O,A)+c(O,B)+c(O,C)}$.
We can represent the dynamics of the particles using independent Poisson processes with a ``thinning'' construction as described in the following paragraphs.

Given a Poisson process $(N^\lambda_T)_{T\in \mathbb R^+}$ of parameter $\lambda$,  we denote by
$\mathcal{T}_{p_i}(N_T^\lambda)$, (with $0\leq p_i\leq 1$ and $\sum_ip_i=1$), the processes obtained by thinning of the process $N^\lambda_T$; a point of $N^\lambda_T$ belongs to $\mathcal T_{p_i}(N^\lambda_T)$ with probability $p_i$ independently from all the other points. The processes  $\mathcal T_{p_i}(N^\lambda_T)$ are independent Poisson processes of parameters $p_i\lambda$.

Our construction is the following. The number of particles injected onto $O$ from boundary site $X$ ($X=A,B,C$) is given by the Poisson process $N^{\phi_Xc(X,O)}_T$ and the three processes are independent.
We perform an independent thinning procedure for each of them obtaining a family of independent Poisson processes $\left(\mathcal T_{p_Y}(N^{\phi_Xc(X,O)})\right)_{X,Y=A,B,C}$. The process $\mathcal T_{p_Y}(N^{\phi_Xc(X,O)})$ represents the number of particles created from the ghost site $X$ and exiting through $Y$ and we can imagine that such particles go directly from $X$ to $Y$ through the edge $(X,Y)$ of the triangle in the top of Figure \ref{star-triangle}. The currents along the edges of the triangle are therefore given by
\begin{equation}\label{civetta}
 j_T(X,Y)=\frac 1T\left[ \mathcal T_{p_Y}(N^{\phi_Xc(X,O)})-\mathcal T_{p_X}(N^{\phi_Yc(X,O)})\right]\,, \qquad X,Y=A,B,C\,,
\end{equation}
and, since all the Poisson processes are independent, we have that the joint large deviation principle for these currents is given by
\begin{multline}
\Gamma_{\phi_A\hat{c}(A,B),\phi_B\hat{c}(B,A)}(j(A,B))+\Gamma_{\phi_B\hat{c}(B,C),\phi_C\hat{c}(C,B)}(j(B,C)) \nonumber \\ +
\Gamma_{\phi_C\hat{c}(C,A),\phi_A\hat{c}(A,C)}(j(C,A))
\end{multline}
with the $\hat{c}(X,Y)$ as in~\eqref{star-tria}.  Finally, the currents along the edges of the star are related to those on the edges of the triangle by \eqref{forcontj} so the right-hand side of \eqref{staruguale} is obtained by the contraction principle.
The proof is finished.
\end{proof}

\subsection{Star-$K_n$ transformations}

Lemma \ref{lemma-a3} can be generalized to a star with $n$ non-central (boundary) nodes labelled $X_1, X_2, \ldots X_n$ each linked to the same central node $O$.  This $n$-ray star can be transformed to an equivalent $n$-complete graph where the central node is removed and each of the $n$ boundary nodes is linked to each of the other $n-1$ nodes. Using a similar argument to the previous subsection, the effective conductivity $\hat{c}(X_i,X_j)$ going from boundary node $X_i$ to boundary node $X_j$ must be given in terms of conductivities in the star by
    \begin{equation*}
      \hat{c}(X_i, X_j)=\frac{ c(X_i,O) c(O,X_j)}{\sum_{i=1}^N c(O,X_i)}.
    \end{equation*}

\subsection{Trace process} \label{trace}
Given $X(t)$ the random walk on $(\Omega_N\cup \partial\Omega_N,E_N)$ with jump rates $(c(x,y))_{(x,y)\in E_N}$ and $S\subset \Omega_N\cup \partial\Omega_N $; we can construct the corresponding trace process $X^S(t)$ which is still a Markov process on $((\Omega_N\cup \partial\Omega_N)\setminus S, E'_N)$ with $E'$ a suitable set of edges. For the general construction and proofs we refer to \cite{L}; here we consider just the case when $S$ is a single node in $\Omega_N$.

Given a node $x\in \Omega_N$ the corresponding trace process $X^x(t)$ is most naturally defined in terms of trajectories. Let $(X(t))_{t\in[0,T]}$ be a trajectory of the original random walk and let $I_i\subseteq [0,T]$ be the time intervals when the walker is at $\{x\}$. The trajectory $\left(X^x(t)\right)_{t\in [0,T-\sum_i|I_i|]}$ of the trace process is obtained cutting out all the periods that the walker spends at $x$ and gluing back together the pieces of trajectories obtained. The process obtained in this way is still Markov and in the case of deletion of one single node the new transition graph and the new rates are exactly the ones obtained by the star-$K_n$ transformations; see \cite{L} for proofs and more details. This also leads to a relation of our construction with harmonic functions which we will not discuss further here.

\section{Large deviations across a cutset and effective conductivity}
We show now how to use the network reduction of the previous section to obtain large deviation rate functionals through a cutset of a general graph.

\subsection{Two sources}

Let us consider a generic graph having just two sites  where particles can be created/annihilated, i.e.\ we have that $\partial \Omega_N$ consists of just two nodes.
The framework is the one in Figure \ref{2-sources} where the two nodes belonging to $\partial \Omega_N$ have been called $L/R$ and marked with a square dot. The edges have left and right conductivities coinciding with the generator weights as in Section~\ref{zrgengen}; to be concrete, $c(x,y)$ represents the conductivity from $x$ to $y$. The parameters $\lambda_L$ and $\lambda_R$ are fixed at the two boundary sites (i.e.\ the two elements of $\partial \Omega_N$) and determine the intensity of creation of particles there.

\begin{figure}
	
	\centering
	
	\includegraphics[width=0.7\textwidth]{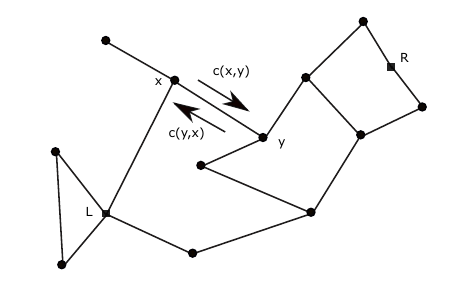}
	\caption{A general graph with two single sources/sinks denoted by $L$ and $R$. Each unoriented bond in the graph is characterized by left and right conductivities as illustrated in the case of the bond $\{x,y\}$.}
        \label{2-sources}
\end{figure}
Let us consider $\tau^+:=\tau^+_{L\cup R}:=\inf\{t>\tau^1 : X(t)\in L\cup R\}$ where $X(\cdot)$ is a Markov chain with transition rates given by $c(x,y)$, $(x,y)\in E_N$. We now define $p^+_L(R,L):=\mathbb P_L\left(X(\tau^+)\in R\right)$, where $\mathbb P_v$ is the Markov measure with initial condition $X(0)=v$.
A similar notation can obviously be used with $L$ and $R$ exchanged. Such probabilities are related to the \emph{capacities} of the sets $L,R$ \cite{L}. Observe that $p^+_L(R,L)$ does not depend on the values $c(R,x)$  and does not change under a global multiplication by a common factor $\alpha$ of all the rates $c(L,x)\to \alpha c(L,x)$. Similar statements are true for $p^+_R(L,R)$.

Let us consider the unoriented graph $(\Omega_N\cup \partial \Omega_N, \mathcal E_N)$ and split the vertices into two connected components: $\Omega^L_N$ containing $L$ and $\Omega^R_N$ containing $R$. We call the collection $\mathcal C$ of unoriented edges $\{x,y\}$, such that $x$ and $y$ belong to different components, an \emph{$L$-$R$ cutset}. For any discrete vector field $j$ such that $\nabla \cdot j=0$ in $\Omega_N$,  we have  by the discrete Gauss-Green theorem that the flux of $j$ across $\mathcal C$ from $L$ to $R$
\begin{equation}\label{flowcut}
j(\mathcal C):=\sum_{\left\{\{x,y\}\in \mathcal C, x\in \Omega^L_N, y\in \Omega^R_N\right\}}j(x,y)
\end{equation}
does not depend on the specific cutset selected. The simplest cutsets are constituted by all the edges having $L$ as a vertex or all the edges having $R$ as a vertex.

We define the empirical flow $j_T(\mathcal C)$ across the cutset $\mathcal C$ as formula~\eqref{flowcut} for the empirical current of \eqref{defempc}. The following result states that, if we are interested in the LDP for $j_T(\mathcal C)$, we can reduce the resistor-like network to one single effective component connecting $L$ and $R$ with effective conductivity from $L$ to $R$ given by $p^+_L(R,L)$ and  effective conductivity from $R$ to $L$ given by $p^+_R(L,R)$; the rate functional for  $j_T(\mathcal C)$  coincides with the cost of this effective single component. Let us define $\hat\lambda_L:=\sum_{x\in \Omega_N}\lambda_Lc(L,x)$ and
$\hat\lambda_R:=\sum_{x\in \Omega_N}\lambda_Rc(R,x)$

\begin{proposition}\label{propf}
Consider a graph with only two boundary vertices, $L$ and $R$, which have associated chemical potentials $\lambda_L$ and $\lambda_R$ respectively. The empirical flow $j_T(\mathcal C)$ across any $L$-$R$ cutset $\mathcal C$ satisfies a large deviation principle when $T\to +\infty$ with rate functional given by
\begin{align}\label{vincenzo}
\mathcal F(q)& = \inf_\phi\inf_{\left\{j: j(\mathcal C)=q\right\}}\sum_{(x,y)\in \mathcal{E}_N}\Gamma_{\phi_xc(x,y),\phi_yc(y,x)}(j(x,y))\nonumber \\
&=\Gamma_{\hat\lambda_Lp^+_L(R,L),\hat\lambda_Rp^+_R(L,R)}(q)\,, \qquad q\in \mathbb R\,,
\end{align}
where the first infimum is over $(\phi_x)_{x\in \Omega_N}$, keeping fixed $\phi_L=\lambda_L$ and $\phi_R=\lambda_R$.
\end{proposition}
\begin{proof}
  We prove the statement using induction on the number of nodes. In the case of a graph with one single internal node we have that the statement of the lemma is obtained with a reduction of two components in series and parallel (as before); the result follows directly by Lemmas \ref{lemma-add} and \ref{testossimoro}.

  We now suppose that the statement is true for any graph with $N$ internal nodes and show that then it is true also for a graph with $N+1$ nodes. The proof follows directly from Lemma \ref{lemma-a3} and the equivalent version for the $k$-star to $k$-complete graph. Consider any internal node in the $(N+1)$-node graph and replace the $k$-star subgraph to which it belongs with the associated $k$-complete graph with corresponding rates. All the divergence-free currents on the two graphs are related by \eqref{forcontj}, or the corresponding formula for a $k$-star graph, and in this zero-divergence case we know that the two currents have the same flow across any cutset. By equality \eqref{staruguale} we have that the infimum over $\phi$ in the first line of \eqref{vincenzo} is the same for the original graph and the reduced one. The proof of the lemma is finished by observing that $\hat\lambda_L p_L^+(R,L)$ and $\hat\lambda_R p_R^+(L,R)$ are the same for the original and the reduced graph since the random walk on the reduced graph is the trace process of the original with a single node removed and a direct coupling argument can be used.  [If the traced-out node is not connected by an edge to a boundary node, then  $p_L^+(R,L)$, $p_R^+(L,R)$, $\hat\lambda_L$, and $\hat\lambda_R$ are all individually unchanged; if the traced-out node is connected to a left (right) boundary node, then in general both $p_L^+(R,L)$ and $\hat\lambda_L$ ($p_R^+(L,R)$ and $\hat\lambda_R$) change but in such a way that their product stays invariant.]
\end{proof}

\subsection{The general case}

The general case can be reduced to the two-source case. Consider a situation like in Figure \ref{2d-1-figura} where we have a general graph in contact with several external sources and we have a cutset that splits the graph and the sources into two components. In Figure \ref{2d-1-figura}, the internal nodes are drawn as round black dots, the external nodes by squared black dots, the unoriented graph as continuous lines and the cutset by a dotted line.

In terms of the current across the cutset, the system is equivalent to a two-source graph with all the external sources on the left side identified with one single left source $L$, as in Figure \ref{2d-2-figura}, and all the sources on the right side identified with a single right source $R$. Let us denote by $\partial \Omega_N^L$ and $\partial \Omega_N^R$ respectively the left and right boundary vertices in the original graph. The chemical potential associated to the left effective boundary point $L$ in Figure \ref{2d-2-figura} can be chosen to be
\begin{equation}
\lambda_L:=\sum_{ \left\{ (x,y)\in E_N, x\in \partial \Omega_N^L\right\} } \lambda_xc(x,y)\,,
\end{equation}
with the chemical potential associated to the right effective boundary point $R$ in Figure \ref{2d-2-figura} correspondingly set as
\begin{equation}
  \lambda_R:=\sum_{ \left\{ (x,y)\in E_N, x\in \partial \Omega_N^R\right\} } \lambda_xc(x,y)\,.
\end{equation}
Each edge $\{x,y\}\in \mathcal E_N$ with $x\in\partial \Omega_N^L$ and $y\in \Omega_N$ can then be substituted by the edge $\{L,y\}$ with conductivities $c(L,y)=\frac{\lambda_xc(x,y)}{\lambda_L}$
and $c(y,L)=c(y,x)$. Of course, there are analogous formulas for the right part too. The problem is now transformed into a two-source problem and we can apply Proposition \ref{propf}, observing that with the choice above we have $\hat\lambda_L=\lambda_L$ and $\hat\lambda_R=\lambda_R$.

\begin{figure}
	
	\centering
	
	\includegraphics[width=0.7\textwidth]{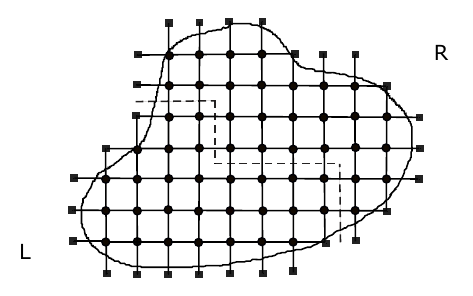}
	\caption{A general graph divided by a cutset of the edges (dashed line); the boundary nodes $\partial \Omega_N$ are drawn as black squares and divided by the cutset into two different classes.}\label{2d-1-figura}
\end{figure}

\begin{figure}
	
	\centering
	
	\includegraphics[width=0.7\textwidth]{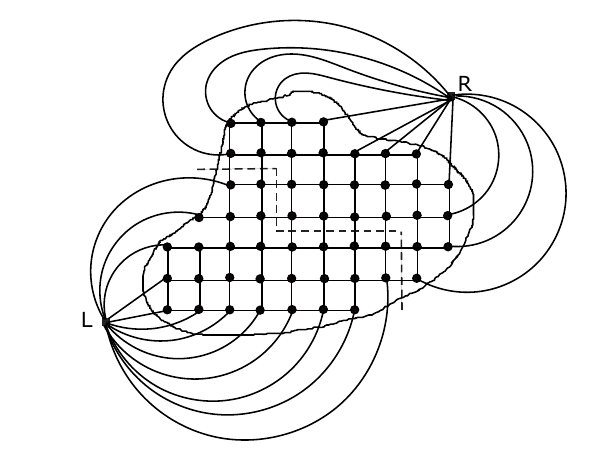}
	\caption{A general graph divided by a cutset of the edges (dashed line) where all the sources/sinks for the two components are identified into two effective sources/sinks, labelled $L$ and $R$ for the left and right components respectively. }\label{2d-2-figura}
\end{figure}

\vskip.2cm
\noindent
{\bf Acknowledgments.} We thank L. Bertini and A. Faggionato for useful discussions.
D.G. acknowledges financial support from the Italian Research Funding Agency (MIUR) through
PRIN project ``Emergence of condensation-like phenomena in interacting particle systems: kinetic and lattice models'', grant no.\ 202277WX43.  R.J.H. benefited from the kind hospitality of the National Institute for Theoretical Physics (NITheP) Stellenbosch at the beginning of this work, and the Stellenbosch Institute for Advanced Study (STIAS) at the end; she is also grateful for support from the London Mathematical Laboratory in the form of an External Fellowship.

\vskip.2cm
\noindent
{\bf Declarations}
The authors have no conflicts of interest. 

\appendix

\section{Large deviations for independent particles}
\label{appA}
In order to have a complete formal large deviation principle for the joint empirical measure and flow in our models of zero-range dynamics with superlinear growth we should verify the validity of several conditions as stated in \cite{BFG,BFG2}. In order to apply the contraction principle, as we do, we should verify in particular the tightness conditions in \cite{BFG2}. We give here an outline of the arguments.

The level 2.5 large deviations in the case of independent particles can be fully proved using the independence of the particles. We consider i.i.d.\ random walks $X_i(t)$ each of which is randomly created at time $\tau_i$ near the boundary (according to the rates of the model), evolves with $N_i$ jumps through the graph, and then disappears on exiting the system at time $\tau_i+T_i$. The times $\tau_i$ at which the random walks are created are Poissonian times of the sources; the independent times $T_i$ the walkers stay in the system are distributed with exponential tails (i.e.\ $\mathbb P(T_i>t)\leq e^{-k t}$, for a suitable positive constant $k$) and the numbers of jumps $N_i$ performed before exiting the system are also independent random variables with exponential tails ($\mathbb P(N_i>n)\leq e^{-K n}$, for another suitable positive constant $K$). From these ingredients, the joint LDP for empirical measure and flow in the case of independent particles, on suitable topology, can be proved using classic statements from i.i.d.\ random variables and Poissonizations.

\begin{remark}\label{remarktilde}
A very useful fact is the following formal statement that can be easily verified.  We consider here a single Poisson process $(N^\lambda_T)_{T\in \mathbb R}$ of parameter $\lambda$, noting that the argument can also be generalized to multiple processes and to marked Poisson processes; it applies, for example, to the LDP for independent particles described above, where the Poisson processes are marked by the random trajectories of the particles.  We let $\tau_i$ be the random events of our Poisson process and now construct $(\tilde N^\lambda_T)_{T\in \mathbb R}$ as the point process that has events at the times $\tau_i+T_i$ where $(T_i)_{\in \mathbb N}$ are i.i.d.\ non-negative random variables independent from $N^\lambda$ and having exponential tails. Then the sequences of random variables $\frac{N^\lambda_T}{T}$ and $\frac{\tilde N^\lambda_T}{T}$ have the same large deviation rate functional for large times.
\end{remark}

Since the zero-range interaction modifies the jump rates but not the probabilities of moving to particular nodes, a superlinear zero-range dynamics can be coupled to the independent-particle process having the same $(\lambda_x)_{x\in \partial \Omega_N}$ and the same $(c(x,y))_{(x,y)\in E_N}$ in such a way that:
\begin{itemize}
	\item particles are created simultaneously at the same random times $\tau_i$ in the two models and receive the same tag: particle number $i$ is created at time $\tau_i$ in both processes;
	
	\item  particles with the same tag in the two processes have exactly the same skeleton of the trajectory inside the system, i.e.\ they make the same jumps even if not at the same times;
	
	\item the particles performing the zero-range interaction jump faster than the free ones (and therefore exit from the graph before times $\tau_i+T_i$).
\end{itemize}
We now briefly justify the second and third points of this coupling.  [The first point should not require justification.]

The second point follows from the fact that we can code the trajectories of the labelled particles using the skeletons and the jump times. Since the interaction affects just the speed of the dynamics, we have that the probability of realization of the skeletons is always (for any function $g$) independent from the jump times.  Moreover, the probabilities of the different skeletons of the different particles are independent and each of them has the law of a discrete-time random walk with transition probabilities from $x$ to $y$ given by ${c(x,y)}/{\sum_{z:(x,z)\in E_N}c(x,z)}$.

The third point can be understood from the following construction, with $g(k)\geq ak$. We associate to each created particle, an internal clock which is a Poisson process of parameter $a$. For the case of independent particles, we let the particles jump, following the path of the independent skeleton, in correspondence to the events of the associated internal clock. On the other hand, for a zero-range dynamics with function $g$, each of the $k$ particles at site $x$ jumps according to a Poisson process with parameter $\frac{g(k)}{k}$.  Hence we must superpose on the internal clock of each particle an independent Poisson clock of parameter $\frac{g(k)-ak}{k}\geq 0$. With this construction, the particles of the zero-range dynamics follow the same skeleton as the independent ones but the zero-range jump times are always before the ones of the independent case.

\smallskip

To prove the joint empirical measure and empirical flow LDP for the superlinear zero-range process we can use the results in \cite{BFG,BFG2} once we have an exponential tightness condition.

The exponential tightness bound on the empirical flow can be obtained by observing that the total number of jumps in the interacting process on any subset of edges of the configuration space is dominated by the number of jumps of a free process for which every new particle created at time $\tau_i$ traverses its path up to the exit point instantaneously and therefore exits again at time $\tau_i$. For this effective independent process the exponential tightness bound can be obtained by classic arguments.

Positive recurrence, absence of explosion, and absence of condensation can all be shown by noting that the number of particles present in the interacting system at each time is always less than or equal to the number in the coupled free process (i.e.\ the original one where particle $i$ evolves in the time interval $(\tau_i, \tau_i+T_i]$).

As we concentrate in this paper on the computational aspects of the LDP rate functional, we do not further discuss the lengthy technical details needed to rigorously establish the validity of the underlying conditions.

\section{Criticality conditions}
\label{appB}

We here show that in the minimization computations in Theorem \ref{Th1} we can restrict consideration to zero-range perturbations of the original zero-range dynamics.

First we fix the measure $\mu=\mu[\phi]$ of the form \eqref{prod-zr} for arbitrary values of the $\phi_x$ and consider all the flows $Q$ such that the pair $(\mu, Q)$ are associated to a fixed current $j$ on the physical space \eqref{def-corrente-fisica}. We then consider the rate functional \eqref{zr-rate-gen} under such constraints, and show that the zero-range perturbation with $F$ the unique solution to equation \eqref{uniqueF} is a critical point. By the convexity of the rate functional \eqref{zr-rate-gen}, we have that this corresponds to the unique minimal point. The subsequent optimization with respect to the parameters $(\phi_x)_{x\in \Lambda_N}$ without constraints is the computation that we do elsewhere in the paper.

To justify this, we start by defining the real function $f(t):=I(\mu[\phi],Q[\phi]+tQ_C)$ where $Q_C$ is a flow associated to the cycle $C$ as in \eqref{QC} and $Q[\phi]$ is the flow \eqref{min-flu-zr} with $F$ given in \eqref{Fsol}. Such changes to the flow argument are generic variations that preserve the divergence-free condition. To preserve also the condition that along an edge $(x,y)$ of the physical lattice there is a fixed current $j(x,y)$, the cycle $C$ on the configuration space has to be such that if it contains an edge corresponding to a jump of one particle from $x$ to $y$ then it must also contain an edge corresponding to a jump from $y$ to $x$. Let $C=(\eta_1,\dots, \eta_{N+1}=\eta_0)$ where
$\eta_{i+1}=\eta_i^{(x_i,y_i)}$. We have therefore
\begin{equation}
f'(0)=\sum_i F(x_i,y_i)=0\,,
\end{equation}
where the last equality follows by the antisymmetry of $F$ and the fact that the number of pairs $(x,y)$ in the cycle $C$ is equal to the number of pairs $(y,x)$.

\smallskip

Second, to conclude that we are indeed obtaining the global minimizer, it remains to prove that the global minimizer $(\mu,Q)$ has a measure $\mu$ of the form \eqref{prod-zr}, i.e.\ is $\mu[\phi]$ for a suitable $\phi$. To do this we show below that there exists a flow $Q[\phi]$, of the form \eqref{min-flu-zr} with $F$ satisfying \eqref{Fsol}, such that the minimum of $I( \mu,Q[\phi])$ over all possible probability measures $\mu$ is obtained at $\mu[\phi]$.

This last fact, together with the previous results, implies that for any current $j$, there exists a critical point of the form $(\mu[\phi],Q[\phi])$ and, by convexity, that is therefore the global minimizer.

For a fixed flow $Q[\phi]$, the stationary conditions in $\mu$ for $I(\mu, Q[\phi])$ computed at $\mu[\phi]$ are
\begin{align}
  \sum_{x\in \partial \Omega_N}\sum_{y:(x,y)\in E_N}&\lambda_xc(x,y)\left(e^{F(x,y)}-1\right) \notag \\
  &+\sum_{x\in \Omega_N} g(\eta_x)\Big[\sum_{y:(x,y)\in E_N}c(x,y)\left(e^{F(x,y)}-1\right)\Big]-k=0\,,\label{ultimaeq}
\end{align}
for any $\eta$, where $k$ is a Lagrange multiplier. If we substitute \eqref{Fsol} in the above equation, we obtain
that for any fixed $x\in \Omega_N$ the quantity in the squared parenthesis in the middle term is
\begin{equation}
\frac{1}{2\phi_x}\left(\sum_{y:(x,y)\in E_N}j(x,y)+\sqrt{j(x,y)^2+4\phi_x\phi_yc(x,y)c(y,x)}\right)-\sum_{y:(x,y)\in E_N}c(x,y)\,,
\end{equation}
which evaluates to zero by the fact that $j$ is divergence-free and condition \eqref{stazd} holds. Since the first term in \eqref{ultimaeq} does not depend on $\eta$, we can fix the Lagrange multiplier $k$ in such a way that the condition is satisfied for any $\eta$.

\section{Proof of Lemma \ref{quadratico}}
\label{apquadratico}

	Squaring twice \eqref{staz-fi} we get
	\begin{equation}\label{staz-qua}
	\left(4\phi_x-\phi_{x+1}-\phi_{x-1}\right)^2-4\phi_{x-1}\phi_{x+1}=4j^2\,,\qquad x=1,\dots,N\,.
	\end{equation}
	An elementary computation shows that any quadratic function \eqref{sol-zr}
	will solve \eqref{staz-qua} provided we have $B^2-4AC=j^2$.
	If we also impose on \eqref{sol-zr}
	the boundary conditions $\phi_0=\lambda_L$ and $\phi_{N+1}=\lambda_R$ we obtain
	the values of the constants in~\eqref{coeffa}.
	
	There is also another possible choice of the coefficients $A$, $B$ and $C$ which is compatible with the boundary values. It corresponds to changing the sign in front of the square root terms in \eqref{coeffa}. This alternative choice has to be rejected since it does not give a solution to \eqref{staz-fi}. Indeed a solution $\{\phi_x\}_{x=1}^N$
	of \eqref{staz-fi} is necessarily non-negative while, with this alternative choice of the sign, the coefficient $C$ is always positive and consequently the function $\phi$ assumes its minimum value
	at its unique critical point $-\frac{B}{2C}$. In this case we have that this critical point is always
	contained in $(0,N+1)$ and moreover the corresponding minimal value is always negative.
	Another elementary computation allows us to show that \eqref{sol-zr} with the coefficients as in
	\eqref{coeffa} is indeed a solution to \eqref{staz-fi}.
	
	It remains to show that any solution to \eqref{staz-fi} is necessarily of the form \eqref{sol-zr}.
	Let us for a moment forget about the boundary value $\phi_{N+1}=\lambda_R$ and consider a solution
	of \eqref{staz-fi} with only the condition $\phi_{0}=\lambda_L$. Since we can write
	$$
	\phi_{x+1}=\frac{\left(4\phi_x-\sqrt{j^2 +4\phi_{x-1}\phi_x}\right)^2-j^2}{4\phi_x}\,, \qquad x=1,\dots ,N\,,
	$$
	we deduce that there is a one-parameter family of solutions depending on the value of $\phi_1$.
	In particular, if we fix the values $\phi_0=\lambda_L$ and $\phi_1=\phi^*_1$ there is at most one solution to
	\eqref{staz-fi}. Our argument is concluded if we show that this unique solution can be written
	in the form \eqref{sol-zr} with $B^2-4AC=j^2$. If we consider \eqref{sol-zr} and impose the conditions
	$\phi_0=\lambda_L$ and $\phi_1=\phi^*_1$ and use the relation $B^2-4AC=j^2$
	we obtain
	\begin{equation}\label{coeff-un}
	\left\{
	\begin{array}{l}
	A=\lambda_L \\
	B= \phi^*_1-\lambda_L-C\\
	C=\lambda_L+\phi^*_1\pm\sqrt{j^2+4\lambda_L\phi^*_1}\,.
	\end{array}
	\right.
	\end{equation}
	The choice of the  $+$ sign in \eqref{coeff-un} does not give a solution to \eqref{staz-fi}.
	Indeed as before in this case the coefficient $C$ is always positive and consequently the function $\phi$ assumes its minimum value at its unique critical point $-\frac{B}{2C}$. This critical point is always
	contained in $(0,N+1)$  and moreover the corresponding minimal value is negative.
	The choice of the  $-$ sign instead gives the unique solution to \eqref{staz-fi}.

\end{document}